\documentclass[journal, onecolumn, 12pt]{IEEEtran}

\ifCLASSINFOpdf
\else
\fi

\usepackage{amssymb}
 \usepackage{amsthm}
 \usepackage{amsmath}
 \usepackage{graphicx}               
 \usepackage{mathrsfs} 
 \usepackage[tight,footnotesize]{subfigure}
 \usepackage{soul}
 \usepackage{booktabs} 
 \usepackage{float}
 \usepackage{xcolor}

\DeclareMathAlphabet{\mathpzc}{OT1}{pzc}{m}{it}
\DeclareMathAlphabet\mathbfcal{OMS}{cmsy}{b}{n}

\newtheorem{lemma}{Lemma}
\newtheorem{definition}{Definition}
\newtheorem{assumption}{Assumption}
\newtheorem{theorem}{Theorem}

\newtheorem{remark}{Remark}

\newcommand{\bunderline}[1]{\underline{#1\mkern-4mu}\mkern4mu }
\newcommand*{\Scale}[2][4]{\scalebox{#1}{$#2$}}%

\begin{document}
%
\title{Robust Cooperative Formation Control of Fixed-Wing Unmanned Aerial Vehicles}
%
%
%

\author{Qingrui Zhang, Hugh H.T. Liu
\thanks{zqrthink@gmail.com}
\thanks{liu@utias.utoronto.ca, Institute for Aerospace Studies, University of Toronto}}

\maketitle

\begin{abstract}
Robust cooperative formation control is investigated in this paper for fixed-wing unmanned aerial vehicles in close formation flight to save energy. A novel cooperative control method is developed. The concept of virtual structure is employed to resolve the difficulty in designing virtual leaders for a large number of UAVs in formation flight. To improve the transient performance, desired trajectories are passed through a group of cooperative filters to generate smooth reference signals, namely the states of the virtual leaders. Model uncertainties due to aerodynamic couplings among UAVs are estimated and compensated using uncertainty and disturbance observers. The entire design, therefore, contains three major components: cooperative filters for motion planning, baseline cooperative control, and uncertainty and disturbance observation. The proposed formation controller could at least secure ultimate bounded control performance for formation tracking. If certain conditions are satisfied, asymptotic formation tracking control could be obtained. Major contributions of this paper lie in two aspects: 1) the difficulty in designing virtual leaders is resolved in terms of the virtual structure concept; 2) a robust cooperative controller is proposed for close formation flight of a large number of UAVs suffering from aerodynamic couplings in between.  The efficiency of the proposed design will be demonstrated using numerical simulations of five UAVs in close formation flight.
\end{abstract}


%
\IEEEpeerreviewmaketitle

\section{Introduction} \label{Sec: Intro}
A fixed-wing UAV flying in close formation, like migratory birds, can reduce their drag and save fuels almost as much as a well designed aerodynamically efficient UAV would bring \cite{Pahle2012AFMC, Bieniawski2014ASM}. In spite of the benefits, close formation flight is challenging for UAV systems. The formation controller of interest must be accurate enough. According to the analysis \cite{Zhang2017JA}, more than $30\%$ of the maximum drag reduction will be lost, if the optimal relative position failed to be maintained within at least $10\%$ wing span accuracy. The formation controller is also expected to be robust enough against all adverse aerodynamic disturbances, such as vortex-induced forces and moments.  The aerodynamic disturbances will either endanger the flight stability of the follower UAV or frustrate close formation flight by deviating the follower UAV away from its optimal position. If a large number of UAVs are of interest to fly in close formation, certain cooperation mechanisms are demanded to mitigate the performance degradation issue in the leader-follower architecture. In a leader-follower control method, follower UAV will track a certain reference trajectory defined based on the states of the leader UAV  \cite{Pachter2001JGCD, Dogan2005JGCD,  Almeida2015GNC}.  Sudden changes in the states of a leader UAV due to either disturbances or mission changes will be reflected in the reference signals of the follower UAV and deteriorate its formation keeping performance.  The performance degradation will be intensified, as it propagates downstream to other followers. Hence, leader-follower control methods are inefficient for a large number of UAVs in close formation. Cooperative control could mitigate the performance degradation issue in a leader-follower method by allowing UAVs in formation to bidirectionally share their information and contribute almost equally to formation keeping.

In cooperative control, UAVs coordinate their actions in terms of both their own states and the states of their neighbours via a bidirectional communication network  \cite{Olfati-Saber2007IEEE, Oh2015Auto,Yang2019Auto}. Intuitively, a cooperative controller performs much better than a leader-follower controller for the formation of a large number of UAVs. So far, different cooperative control strategies have been proposed, such as the potential field method \cite{Dimarogonas2006Auto, Dimarogonas2007TAC}, distance-based formation control \cite{Oh2015Auto, Sun2016IJC, Deghat2016TAC}, behaviour-based control \cite{Lawton2003TRA}, virtual structure based control \cite{Ren2004JGCD, Sadowska2011IJC,  Rezaee2014TIE}, virtual leader-based control \cite{Egerstedt2001TAC, Dong2017TIE}, and consensus-based control \cite{Lafferriere2005SCL, Wang2013TCST, Dong2016CEP},  etc. 

In the potential filed method, artificial potential functions are designed to characterize the interactions among vehicles \cite{Dimarogonas2006Auto, Dimarogonas2007TAC}. If two vehicles are close to each other, an repulsive force is produced, while an attractive force is generated if two vehicles are far away from each other.  The traditional potential field method cannot secure an unique and unambiguous formation shape. This drawback could be resolved by adding more constraints to the communication network, resulting in the so-called rigid formation and accordingly the distance-based formation control \cite{Oh2015Auto}.  Construction of a rigid graph, which is the major concern for the distance-based formation control, will become formidable with the increase of the number of vehicles. More importantly, the distance-based formation control doesn't account for formation rotation that happens for close formation flight under different maneuvers. In the behaviour-based control, several competing objectives are defined, such as tracking desired positions, avoiding collision, and holding relative positions to neighbour vehicles. The first step is to design the desired formation pattern, including the desired formation shape and location,  thereby generating the next desired waypoint for each vehicle  \cite{Balch1998TRA, Lawton2003TRA}. The second step is to steer all vehicles to their desired position in terms of a navigation law, and meanwhile, vehicles cooperate with their neighbours to keep the required formation shape  \cite{Balch1998TRA, Lawton2003TRA}. The stability of behaviour-based control is hard to analyze mathematically, and formation keeping accuracy cannot be guaranteed during maneuvers. An modification to the behaviour-based control is the virtual structure method in which the formation is modelled as a rigid body called a virtual structure \cite{Sadowska2011IJC}. The rigid body is inscribed in a circle, on which all vehicles are located \cite{Rezaee2014TIE}. The motion of the virtual structure is provided. The desired motion for a vehicle in formation is thereafter determined from its desired location in the virtual structure.  The virtual structure-based method is only applied to polygon formation, as it is too complex to describe a general formation shape using a virtual structure.  Another modification to the  behaviour-based control is the virtual leader-based control which introduces a group of virtual leaders to describe both the desired formation shape and motions of all vehicles. In comparison with the virtual structure-based method, the virtual leader-based control is more flexible in the characterization of formation shapes.  The cost is that it will get more complex to simultaneously describe the motions of all virtual leaders with the increase of the vehicle number.  

This paper investigates the cooperative control problem of  close formation flight of a large number of UAVs.  The existing research on close formation flight  is mainly interested in the case of two or three UAVs, so a leader-follower method is preferred to a cooperative control due to its simplicity \cite{Pachter2001JGCD, Gu2006TCST, Chichka2006JGCD, Brodecki2015JGCD, Zhang2018AESCTE, Zhang2018TIE}. Different from existing research, we consider more than three UAVs in close formation flight, so a cooperative controller is more promising. Additionally, close formation flight requires all UAVs to fly in a ``V-shape'' formation in order to maximize the aerodynamic benefits. With the consideration of the merits and limits of different cooperative formation controllers, the virtual leader-based method is preferred for close formation flight. However, as aforementioned, the simultaneous design of virtual leaders will be very complex for formation flight of a large number of UAVs. To solve this issue, the virtual structure concept is borrowed in the design of virtual leaders. The entire formation is characterized as a rigid body. The desired formation trajectory is defined on the geometric center of the rigid body.  The desired trajectory of a UAV is determined according to its required relative position to the formation center. The virtual leaders are obtained by passing all desired trajectories through certain cooperative filters which were employed to improve the transient performance. Based on the proposed cooperative controller,  uncertainty and disturbance observers are introduced to estimate and compensate model uncertainties induced by trailing vortices. The proposed cooperative formation controller could at least ensure ultimate bounded control performance for formation tracking. Furthermore, if certain conditions are satisfied, asymptotic formation tracking control will be obtained. Major contributions of this paper lie in two aspects: 1) a robust cooperative controller is proposed for close formation flight of a large number of UAVs; 2) the difficulty in designing virtual leaders is resolved in terms of the virtual structure concept.  Numerical simulations are presented to demonstrate the efficiency of the proposed design.

The rest of the paper is organized as follows. In Section \ref{Sec: Preli}, some preliminaries are provided. Section \ref{Sec: ProbForm} formulates the major problem. Virtual leader design is presented in Section \ref{Sec: FormTraj}. The robust cooperative control design is described in Section  \ref{Sec: CntrlDesign}, while its stability is analyzed in Section \ref{Sec: StabAnal}. In Section \ref{Sec: NumSim}, numerical simulations are reported. Conclusions are given in Section \ref{Sec: Conclu}.

\section{Preliminaries} \label{Sec: Preli}

The communication topology among $n$ UAVs in close formation flight is modeled using a undirected graph $\mathscr{G}$. Some necessary knowledge is reviewed, and for more details on graph theory, an reader could refer to \cite{Diestel2000Book, Godsil2001Book}. The undirected graph $\mathscr{G}$ is denoted by a triplet $\mathscr{G}:=\left\{\mathscr{V}\text{, }\mathscr{E}\text{, }\boldsymbol{\mathcal{A}}\right\}$ with a node set $\mathscr{V}=\left\{1\text{, }\ldots\text{, } n\right\}$, an edge set $\mathscr{E}\subseteq \mathscr{V}\times\mathscr{V}$, and an adjacency matrix $\boldsymbol{\mathcal{A}}=\left[a_{ij}\right]\in\mathbb{R}^{n\times n}$. For each node $i\in\mathscr{V}$, it represents a UAV $i$ in close formation.  If a UAV $i$ is able to receive information from a UAV $j$ ($j\neq i$), there exists an edge $\left(i\text{, } j\right)\in\mathscr{E}$ and accordingly, $a_{ij}=1$. Furthermore, if $\left(i\text{, } j\right)\in\mathscr{E}$,  UAV $j$ is called an neighbour of UAV $i$. The neighbourhood of UAV $i$ is denoted by $\mathscr{N}_i:=\left\{\forall j\in \mathscr{V} \vert \left(i\text{, } j\right)\in\mathscr{E}\right\}$. In a undirected communication topology, the communication is bidirectional. If $\left(i\text{, } j\right)\in\mathscr{E}$, there exists $\left(j\text{, } i\right)\in\mathscr{E}$, namely $a_{ij}=a_{ji}=1$, whereas $a_{ij}=a_{ji}=0$ if $\left(i\text{,}j\right)\not\in\mathscr{E}$. Note that there always exists $\left(i\text{, } i\right)\not\in\mathscr{E}$, meaning $a_{ii}=0$, $\forall i\in \mathscr{V}$. Hence, $\boldsymbol{\mathcal{A}}$ is a symmetric matrix with zero diagonal elements.   The degree matrix of a graph $\mathscr{G}$ is a diagonal matrix $\boldsymbol{\mathcal{D}}=diag\left\{\mathpzc{d_1}\text{, }\ldots\text{, }\mathpzc{d_n}\right\}$, where $\mathpzc{d_i}=\sum_{j=1}^{n}a_{ij}$, $\forall i\in\mathscr{V}$. The Laplacian matrix of $\mathscr{G}$ is defined to be $\boldsymbol{\mathcal{L}}=\boldsymbol{\mathcal{D}}-\boldsymbol{\mathcal{A}}$. A path on $\mathscr{G}$ between $i_1$ and $i_l$ is a sequence of edges of the form $(i_j$, $i_{j+1})$, where $j=1$, $\ldots$, $l-1$ and $i_j\in\mathscr{V}$. A undirected graph $\mathscr{G}$ is connected, if there exists a path from each node $i\in\mathscr{V}$ to any other nodes. For a undirected graph $\mathscr{G}$, the following lemma exists \cite{Merris1994LAA}.
\begin{lemma} \label{Lem: LaplacianMatrix}
The undirected graph $\mathscr{G}$ is connected, if and only if $0$ is a simple eigenvalue of the Laplacian matrix $\boldsymbol{\mathcal{L}}$ with the associate eigenvector $\mathbf{1}_{n}\in\mathbb{R}^{n \times1}$ of all ones, while all the other eigenvalues are positive.
\end{lemma}
The graph $\mathscr{G}$ is assumed to be connected in the design, so the Laplacian matrix $\boldsymbol{\mathcal{L}}$ is positive semi-definite, and $\boldsymbol{\mathcal{L}}\mathbf{1}_{n}=\mathbf{0}$ according to Lemma \ref{Lem: LaplacianMatrix}.

\begin{lemma} \label{Lem: LaplacianMat_Property}
Let $\boldsymbol{\mathcal{L}}$ be a Laplacian matrix of a undirected graph $\mathscr{G}$ with $n$ nodes, and $\boldsymbol{\mathcal{B}}=diag\left\{\mathpzc{b_1}\text{, }\ldots\text{, }\mathpzc{b_n}\right\}$ be a positive semi-definite diagonal matrix with $Rank\left(\boldsymbol{\mathcal{B}}\right)\geq1$. If $\mathscr{G}$  is connected,  $\boldsymbol{\mathcal{L}}+\boldsymbol{\mathcal{B}}$ will be positive definite, namely $\lambda_{min}\left(\boldsymbol{\mathcal{L}}+\boldsymbol{\mathcal{B}}\right)>0$ where $\lambda_{min}\left(\cdot\right)$ is the minimal eigenvalue of a matrix. Furthermore,  suppose $\boldsymbol{\mathcal{C}}_1\in\mathbb{R}^{m\times m}$ and $\boldsymbol{\mathcal{C}}_2\in\mathbb{R}^{m\times m}$ are two positive definite diagonal matrices.  Then $\boldsymbol{\mathcal{L}}\otimes\boldsymbol{\mathcal{C}}_1+\boldsymbol{\mathcal{B}}\otimes\boldsymbol{\mathcal{C}}_2$ will be positive definite. 
\end{lemma}
\begin{proof}
The Laplacian matrix $\boldsymbol{\mathcal{L}}$ of  a undirected graph $\mathscr{G}$ is symmetric.  If $\mathscr{G}$ is connected, $\boldsymbol{\mathcal{L}}$ will be positive semi-definite according to Lemma \ref{Lem: LaplacianMatrix}. The kernel of $\boldsymbol{\mathcal{L}}$ is $span\left(\mathbf{1}_n\right)$ where $span\left(\cdot\right)$ is a linear span function.  Since both $\boldsymbol{\mathcal{L}}$ and $\boldsymbol{\mathcal{B}}$ are positive semi-definite,  $\boldsymbol{\mathcal{L}}+\boldsymbol{\mathcal{B}}$ must be positive semi-definite, namely $\lambda_{min}\left(\boldsymbol{\mathcal{L}}+\boldsymbol{\mathcal{B}}\right)\geq0$. In what follows, we will show that $\lambda_{min}\left(\boldsymbol{\mathcal{L}}+\boldsymbol{\mathcal{B}}\right)>0$ or $\lambda_{min}\left(\boldsymbol{\mathcal{L}}+\boldsymbol{\mathcal{B}}\right)\neq0$. Let $\boldsymbol{\nu}\in\mathbb{R}^{n\times 1}$ be an eigenvector of $\boldsymbol{\mathcal{L}}+\boldsymbol{\mathcal{B}}$ with $\boldsymbol{\nu}^T\boldsymbol{\nu}=1$ and  $\lambda\boldsymbol{\nu}=\left(\boldsymbol{\mathcal{L}}+\boldsymbol{\mathcal{B}}\right)\boldsymbol{\nu}$. The minimal eigenvalue is
\begin{equation*}
\lambda_{min}\left(\boldsymbol{\mathcal{L}}+\boldsymbol{\mathcal{B}}\right)=\min_{\boldsymbol{\nu}}\boldsymbol{\nu}^T\left(\boldsymbol{\mathcal{L}}+\boldsymbol{\mathcal{B}}\right)\boldsymbol{\nu}=\min_{\boldsymbol{\nu}}\left(\boldsymbol{\nu}^T\boldsymbol{\mathcal{L}}\boldsymbol{\nu}+\boldsymbol{\nu}^T\boldsymbol{\mathcal{B}}\boldsymbol{\nu}\right)
\end{equation*}
where $\boldsymbol{\nu}^T\boldsymbol{\mathcal{L}}\boldsymbol{\nu}\geq 0$ and $\boldsymbol{\nu}^T\boldsymbol{\mathcal{B}}\boldsymbol{\nu}\geq 0$, as  both $\boldsymbol{\mathcal{L}}$ and $\boldsymbol{\mathcal{B}}$ are positive semi-definite. Therefore, $\lambda_{min}\left(\boldsymbol{\mathcal{L}}+\boldsymbol{\mathcal{B}}\right)=0$, if and only if $\boldsymbol{\nu}^T\boldsymbol{\mathcal{L}}\boldsymbol{\nu}= 0$ and $\boldsymbol{\nu}^T\boldsymbol{\mathcal{B}}\boldsymbol{\nu}=0$.  

We thereafter divide $\mathbb{R}^{n\times 1}$ into two subspaces, $span\left(\mathbf{1}_n\right)$ and the complement of $span\left(\mathbf{1}_n\right)$. When $\boldsymbol{\nu}\in span\left(\mathbf{1}_n\right)$, there exist $\boldsymbol{\nu}^T\boldsymbol{\mathcal{L}}\boldsymbol{\nu}= 0$ and $\boldsymbol{\nu}^T\boldsymbol{\mathcal{B}}\boldsymbol{\nu}=\frac{1}{n}\sum_{i=1}^{n}\mathpzc{b_i}>0$. When  $\boldsymbol{\nu}\not\in span\left(\mathbf{1}_n\right)$, $\boldsymbol{\nu}^T\boldsymbol{\mathcal{L}}\boldsymbol{\nu}\neq 0$, as $\boldsymbol{\nu}$ is not at the kernel space of $\boldsymbol{\mathcal{L}}$. With the consideration of the positive semi-definite property of $\boldsymbol{\mathcal{L}}$, $\boldsymbol{\nu}^T\boldsymbol{\mathcal{L}}\boldsymbol{\nu}> 0$, if $\boldsymbol{\nu}\not\in span\left(\mathbf{1}_n\right)$. In summary, $\boldsymbol{\mathcal{L}}+\boldsymbol{\mathcal{B}}$ satisfies the following two properties.
\begin{enumerate}
\item $\boldsymbol{\mathcal{L}}+\boldsymbol{\mathcal{B}}$ is positive semi-definite, and
\item $\lambda_{min}\left(\boldsymbol{\mathcal{L}}+\boldsymbol{\mathcal{B}}\right)\neq0$.
\end{enumerate}
Hence, one is able to conclude that $\boldsymbol{\mathcal{L}}+\boldsymbol{\mathcal{B}}$ will be positive definite.

Since $\boldsymbol{\mathcal{C}}_1$ and $\boldsymbol{\mathcal{C}}_2$ are two diagonal matrices, they have a common eigenvector space which is denoted by $\mathbfcal{Q}$. The kernel of $\boldsymbol{\mathcal{L}}\otimes\boldsymbol{\mathcal{C}}_1$ is $span\left(\mathbf{1}_n\otimes \mathbfcal{Q}\right)$. For a vector $\boldsymbol{\nu}\in span\left(\mathbf{1}_n\otimes \mathbfcal{Q}\right)$, it is easy to know that $\left(\boldsymbol{\mathcal{B}}\otimes\boldsymbol{\mathcal{C}}_2\right)\boldsymbol{\nu}\neq\mathbf{0}$, while $\left(\boldsymbol{\mathcal{L}}\otimes\boldsymbol{\mathcal{C}}_1\right)\boldsymbol{\nu}\neq 0$ for any vector $\boldsymbol{\nu}\not\in span\left(\mathbf{1}_n\otimes \mathbfcal{Q}\right)$. Following similar analysis process in proving the positive definiteness of $\boldsymbol{\mathcal{L}}+\boldsymbol{\mathcal{B}}$, we could reach the conclusion that $\boldsymbol{\mathcal{L}}\otimes\boldsymbol{\mathcal{C}}_1+\boldsymbol{\mathcal{B}}\otimes\boldsymbol{\mathcal{C}}_2$ is positive definite.
\end{proof}
\section{Problem formulation} \label{Sec: ProbForm}
 The cooperative formation controller performs as an outer-loop controller, while the inner-loop attitude dynamics are assumed to be stabilized by a certain inner-loop controller. Assume the sideslip angle of a UAV is stabilized to be zero, while the angle of attack is kept to be small. A six-degree-of-freedom (6DoF) nonlinear UAV model is used as given in (\ref{Eq: SysDyn_Coop}).
\begin{equation}
\left\{\begin{array}{ccl}
\dot{x}_i&=&V_{Ti} \cos{\gamma_{i}}\cos{\psi_i} \\
\dot{y}_i&=&V_{Ti} \cos{\gamma_{i}}\sin{\psi_i} \\ 
\dot{z}_i&=&-V_{Ti} \sin{\gamma_{i}} \\
\dot{V}_{gi}&=& {\left(T_i -D_i\right)}/{m_i}-g \sin{\gamma_{i}}+d_{Vi} \\
\dot{\gamma}_{ai}&=& {L_i\cos{\mu_i}}/{\left(m_iV_{Ti} \right)}-{g\cos{\gamma_{i}}}/{V_{Ti} } +d_{\gamma i} \\
\dot{\psi}_i&=&{L_i\sin{\mu_i}}/{\left(m_iV_{Ti} \cos{\gamma_{i}}\right)}+d_{\psi i}
\end{array}\right. \qquad  \left(i\in\mathscr{V}\right) \label{Eq: SysDyn_Coop}
\end{equation}
where $x_i$, ${y}_i$, and ${z}_i$ denote the position coordinates of UAV $i$ in the inertial frame (the north-east-down frame, NED), $V_{Ti}$ is the total speed of a UAV in close formation, which is the resultant speed of the airspeed and the trailing vortex-induced wake velocity, $\gamma_{i}$ and $\psi_i$ are the flight path angle, and course angle of UAV $i$ in the trailing vortices of an upstream UAV, respectively,  $m_i$ is the mass, $g$ is the gravity acceleration, $T_i$ is the engine thrust, $L_i$ is the lift, $D_i$ is the drag, $\mu_i$ is the bank angle, and $d_{Vi}$, $d_{\gamma i}$, and $d_{\chi i}$ are lumped terms of model uncertainties and trailing vortex-induced disturbances with $d_{Vi} = -\frac{\Delta D_i}{m_i}$,  $d_{\gamma i} =  \frac{\Delta L_i \cos{\mu_i}- \left(Y_i+\Delta Y_i\right)\sin{\mu_i}}{m_iV_{Ti} }$, and 
$d_{\psi i} =  \frac{\Delta L_i \sin{\mu_i}+ \left(Y_i+\Delta Y_i\right)\cos{\mu_i}}{m_iV_{Ti} \cos{\gamma_{i}}}$, 
where $\Delta D_i$, $\Delta L_i$, and $\Delta Y_i$ are trailing vortices-induced drag, lift and side force, respectively, and $Y_i$ is the side force treated as a model uncertainty. To simplify the design process, $D_i$ is assumed to be known, but it can be taken as an unknown term in real implementations. The flight path angle $\gamma_{i}$ and course angle ${\psi_i}$ are computed by  $\sin {\gamma}_{i}=-\frac{\dot{z}_i}{V_{Ti}}$ and $\tan{\psi_i} =\frac{\dot{y}_i}{\dot{x}_i}$, respectively. 
Control inputs for (\ref{Eq: SysDyn_Coop}) are chosen to be $T_i$, $\alpha_i$, and $\mu_i$. 
Differentiating $x_i$, ${y}_i$, and ${z}_i$ with respect to time twice yields
\begin{equation}
\ddot{x}_i=u_{x i}+d_{xi} \text{,}\qquad \ddot{y}_i=u_{y i}+d_{yi} \text{,}\qquad \ddot{z}_i=u_{z i}+d_{zi}  \qquad\left(i\in\mathscr{V}\right)\label{Eq: DoubleIntegator}
\end{equation}
where $i=1$, $2$, $\ldots$, $n$, $u_{x i}$, $u_{y i}$, and $u_{z i}$ are new control variables given in (\ref{Eq: CntrlConvert2}), and $d_{x i}$, $d_{y i}$, and $d_{z i}$ are are uncertainty and disturbance terms given in (\ref{Eq: UncerDist_2ndSys}). 
\begin{eqnarray}
\left\{\begin{array}{lcl}
u_{xi} &=& u_{Vi}\cos \gamma_{i} \cos \psi_i -u_{\gamma i}V_{Ti} \sin \gamma_{i} \cos \psi_i  -u_{\psi i}V_{Ti} \cos \gamma_{i}\sin \psi_i \\
u_{yi} &=&  u_{Vi}\cos \gamma_{i} \sin \psi_i  -u_{\gamma i}V_{Ti} \sin \gamma_{i} \sin \psi_i +u_{\psi i}  V_{Ti} \cos \gamma_{i}\cos \psi_i\\
u_{zi} &=& -u_{Vi}\sin \gamma_{i} -u_{\gamma i}V_{Ti} \cos \gamma_{i}
\end{array}
\right. \label{Eq: CntrlConvert2}
\\
\left\{\begin{array}{lcl}
d_{xi} &=& d_{Vi}\cos \gamma_{i} \cos \psi_i -d_{\gamma i}V_{Ti} \sin \gamma_{i} \cos \psi_i  -d_{\psi i}V_{Ti} \cos \gamma_{i}\sin \psi_i \\
d_{y i} &=&  d_{Vi}\cos \gamma_{i} \sin \psi_i  -d_{\gamma i}V_{Ti} \sin \gamma_{i} \sin \psi_i +d_{\psi i}  V_{Ti} \cos \gamma_{i}\cos \psi_i\\
d_{z i} &=& -d_{Vi}\sin \gamma_{i} -d_{\gamma i}V_{Ti} \cos \gamma_{i}
\end{array}
\right. \label{Eq: UncerDist_2ndSys}
\end{eqnarray}
where $u_{Vi} = \frac{T_i -D_i}{m_i}-g \sin{\gamma_{i}} $, $u_{\gamma i} =   \frac{L_i\cos{\mu_i}}{m_iV_{Ti} }-\frac{g\cos{\gamma_{i}}}{V_{Ti} }$, and $u_{\psi i} =  \frac{L_i\sin{\mu_i}}{m_iV_{Ti} \cos{\gamma_{i}}}$.
For the sake of stability analysis, the following assumption is introduced.
\begin{assumption} \label{Assump: UncerDist}
Both $d_{x i}$, $d_{y i}$, $d_{z i}$ and their derivatives are bounded, namely $\vert d_{x i} \vert \leq  \bar{d}_{xi}$, $\vert d_{y i} \vert \leq  \bar{d}_{yi}$, $\vert d_{z i}\vert \leq  \bar{d}_{zi}$, $\vert \dot{d}_{x i} \vert \leq  \bar{\dot{d}}_{xi}$, $\vert \dot{d}_{y i} \vert \leq  \bar{\dot{d}}_{yi}$, and $\vert \dot{d}_{z i}\vert \leq  \bar{\dot{d}}_{zi}$.
\end{assumption}
In the main results $u_{x i}$, $u_{y i}$, and $u_{z i}$ will be designed based on the double-integrator model (\ref{Eq: DoubleIntegator}). Once $u_{x i}$, $u_{y i}$, and $u_{z i}$ are obtained, $u_{Vi}$, $u_{\gamma i}$, and $u_{\psi i}$ are computed by 
\begin{equation}
\left\{\begin{array}{lcl}
u_{Vi} &=& u_{xi}\cos \gamma_{i} \cos \psi_i +u_{y i}\cos \gamma_{i} \sin \psi_{i}  -u_{z i}\sin \gamma_{i}\\
u_{\gamma i} &=&  -\frac{u_{xi}}{V_{Ti} }\sin \gamma_{i} \cos \psi_i  -\frac{u_{y i}}{V_{Ti} } \sin \gamma_{i} \sin \psi_i -\frac{u_{z i}}{V_{Ti} }\sin \gamma_{i}\\
u_{\psi i} &=& -\frac{u_{xi}\sin \psi_{i} }{V_{Ti} \cos \gamma_{i}} +\frac{u_{y i}\cos{\psi_i}}{V_{Ti} \cos \gamma_{i}}
\end{array}
\right. \label{Eq: CntrlConvert2_Inverse}
\end{equation}
The real control inputs $T_i$, $L_i$, and $\mu_i$ are thus calculated using
\begin{equation}
\left\{\begin{array}{lcl}
T_{i} &=& m_i u_{Vi}+m_ig\sin{\gamma_{i}}+D_i\\
L_{i} &=& \sqrt{\left(m_iV_{Ti} u_{\gamma i}+m_ig\cos{\gamma_{i}}\right)^2+\left(m_iV_{Ti} u_{\psi i}\cos{\gamma_{i}}\right)^2}\\
\mu_{i} &=& \tan^{-1}\left(\frac{m_iV_{Ti} u_{\psi i}\cos{\gamma_{i}}}{m_iV_{Ti} u_{\gamma i}+m_ig\cos{\gamma_{i}}}\right)
\end{array}
\right. \label{Eq: CntrlConvert_Inverse}
\end{equation}
\begin{figure}[tbp]
\centering
\includegraphics[width=0.8\textwidth]{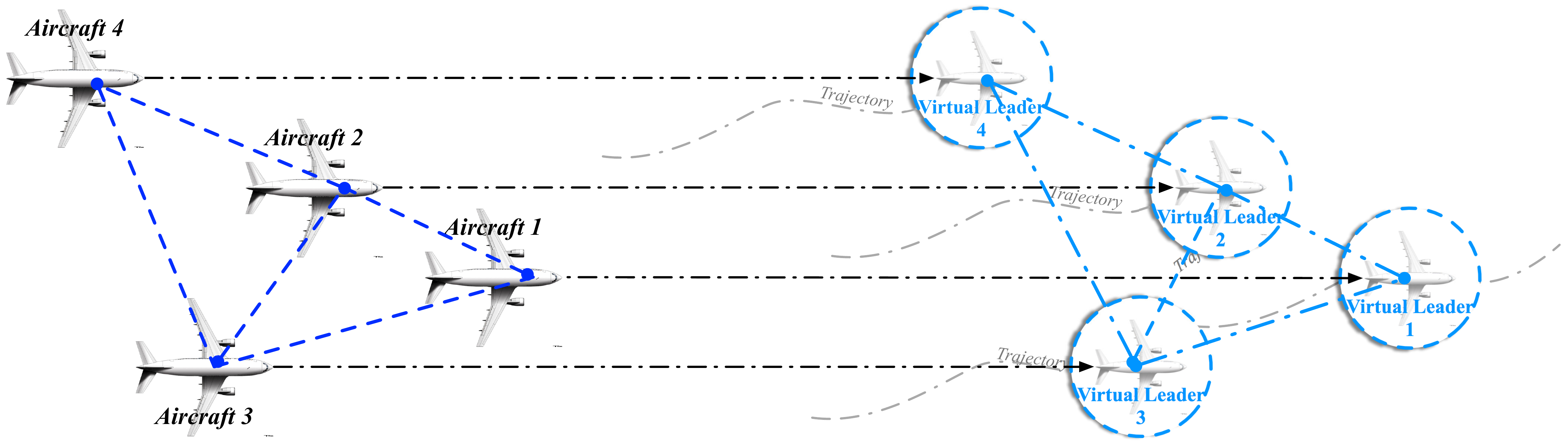}\\ \vspace{-3mm}
\caption{Virtual leader-based close formation flight control}\vspace{-3mm}
\label{Fig: VL_Topology}
\end{figure}\noindent
The control input for $\alpha_i$ is obtained by $\alpha_i=\frac{L_i-L_{0i}}{L_{\alpha i}}$.  Let $\mathbf{p}_i=\left[x_i\text{, }y_i\text{, } z_i\right]^T$. Expressed in a compact form, the double integrator system is
\begin{equation}
\mathbf{\dot{p}}_i=\mathbf{v}_i \text{, }\qquad \mathbf{\dot{v}}_i=\mathbf{u}_i+\mathbf{d}_i \qquad  \left(i\in\mathscr{V}\right)\label{Eq: DoubleInteg_Mat}
\end{equation}
where $\mathbf{v}_i=\left[\dot{x}_i\text{, }\dot{y}_i\text{, } \dot{z}_i\right]^T$ is the speed vector, $\mathbf{u}_i=\left[u_{x i}\text{, }u_{y i}\text{, }u_{z i}\right]^T$ is the control input vector, and  $\mathbf{d}_i=\left[d_{x i}\text{, }d_{y i}\text{, }d_{z i}\right]^T$ is the model uncertainty and disturbance vector. The objective is to coordinate UAVs to track formation trajectories defined by a group of virtual leaders as shown in Figure \ref{Fig: VL_Topology}, and meanwhile, UAVs will share tracking errors with their neighbours. Let $\mathbf{r}_i$ be the position of the virtual leader $i$, and assume that both the first and second derivatives of $\mathbf{r}_i$ are available. Two definitions are introduced.
\begin{definition} \label{Def: Formation_VL}
Asymptotic formation flight is achieved, if for each UAV with any initial states $\mathbf{p}_i\left(0\right)$ and $\mathbf{v}_i\left(0\right)$, there exist
\begin{equation}
\lim_{t\to\infty} \Vert\mathbf{p}_i-\mathbf{r}_i\Vert_2 =0 \qquad \text{and}\qquad\lim_{t\to\infty} \Vert\mathbf{v}_i-\mathbf{\dot{r}}_i\Vert_2 =0 \qquad \left(i\in\mathscr{V}\right) \label{Eq: Formation_VL}
\end{equation}
\end{definition}
\begin{definition} \label{Def: Practical_Formation_VL}
Bounded formation flight is achieved, if for each UAV with any initial conditions $\mathbf{p}_i\left(0\right)$ and $\mathbf{v}_i\left(0\right)$, there exist
\begin{equation}
\lim_{t\to\infty} \Vert\mathbf{p}_i-\mathbf{r}_i\Vert_2 \leq \epsilon_p \qquad \text{and}\qquad\lim_{t\to\infty} \Vert\mathbf{v}_i-\mathbf{\dot{r}}_i\Vert_2 \leq \epsilon_v  \qquad \left(i\in\mathscr{V}\right) \label{Eq: Practical_Formation_VL}
\end{equation}
where $\epsilon_p$ and $\epsilon_v$ are positive scalars that can be made arbitrarily small by design parameters.
\end{definition}
\section{Cooperative motion planning} \label{Sec: FormTraj}
\begin{figure}[tbp]
\centering
\subfigure[Reference motions]{\includegraphics[width=0.45\textwidth]{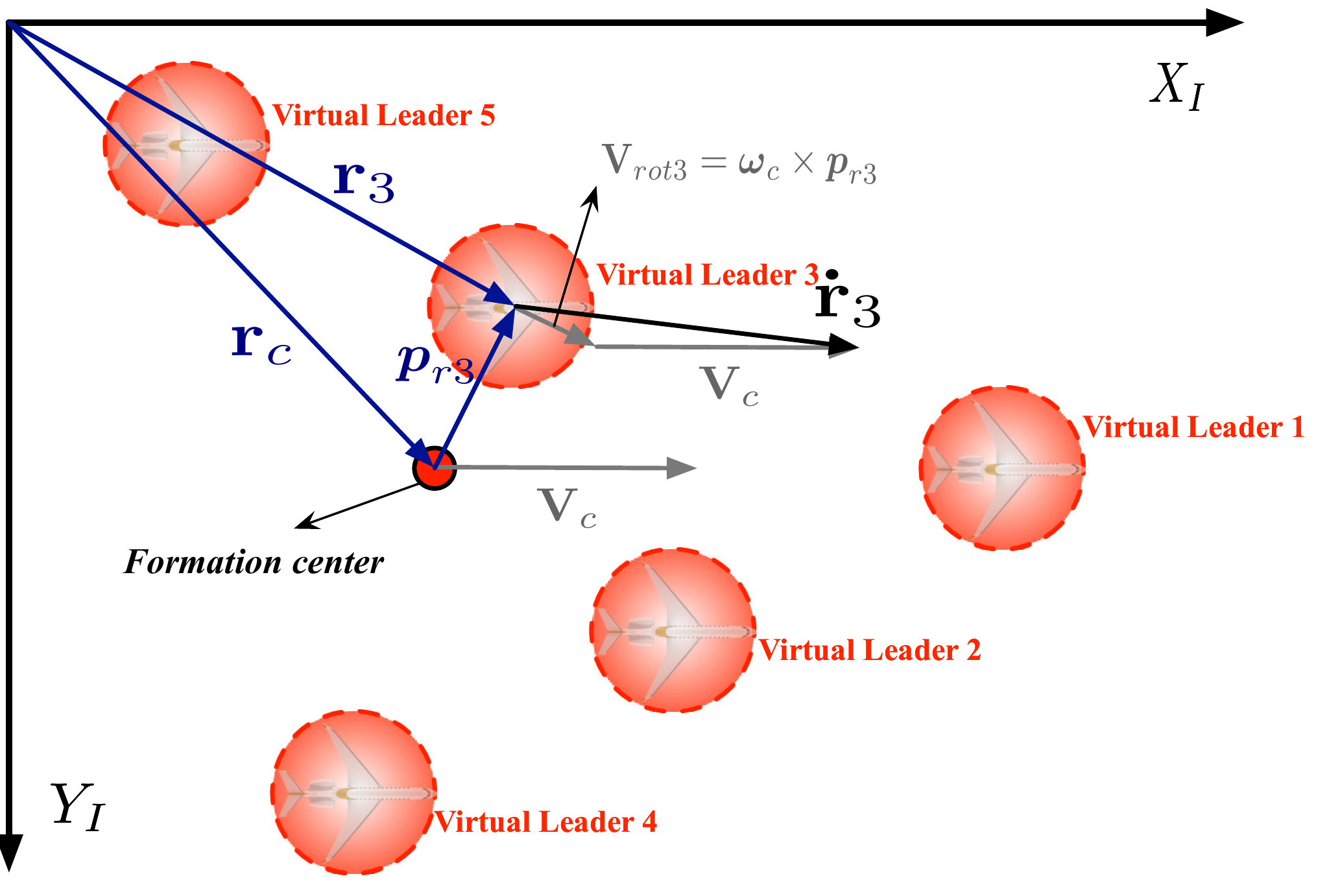}}\hfill
\subfigure[Reference positions in a local frame]{\includegraphics[width=0.45\textwidth]{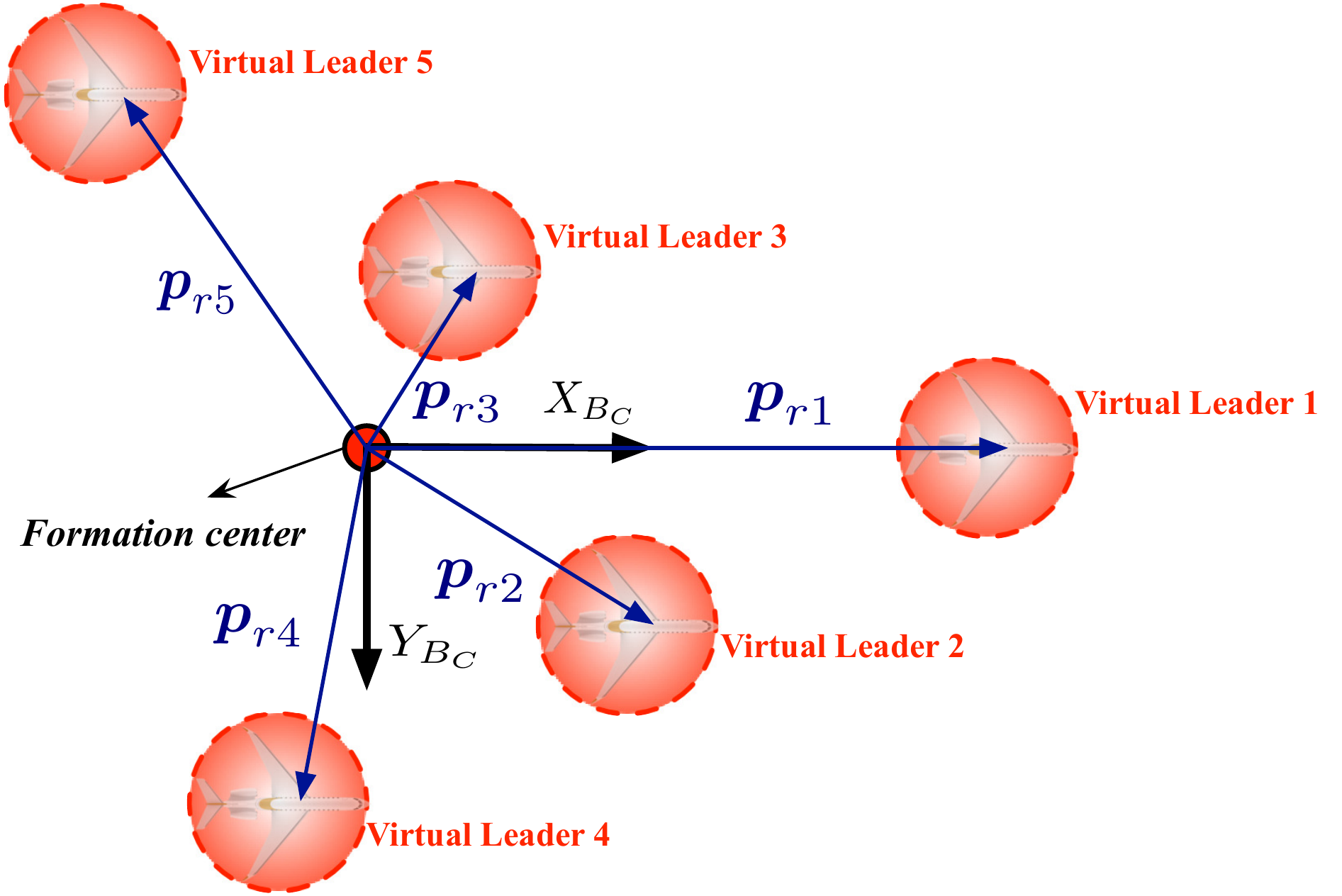}}\\ \vspace{-3mm}
\caption{ Motions and coordinates of virtual leaders}\vspace{-3mm}
\label{Fig: VL_Motions}
\end{figure}
The challenge in a virtual leader-based method is to simultaneously design virtual leaders and make them keep the optimal formation shape under different maneuvers. Such a challenge will get much more difficult, when a large number of UAVs are considered in formation flight. In this paper, the virtual structure concept is employed, where desired formation shape is taken as a certain rigid body. The formation reference  trajectory is described on the geometric center of the rigid body. As  shown in Figure \ref{Fig: VL_Motions}, desired motions for a UAV are defined using its relative position and motion to the geometric center of the rigid body.  In Figure \ref{Fig: VL_Motions}.a, $\mathbf{r}_{i}$ is the desired position vector of UAV $i$ in the inertial frame, whereas $\mathbf{r}_c$ is the position vector of the formation center in the inertial frame. In close formation, all UAVs fly in the same altitude. Horizontal relative positions are constant in a local frame fixed to the formation center. Shown in Figure \ref{Fig: VL_Motions}.b is the constant position vector  $\boldsymbol{p}_{ri}$ of a virtual leader $i$ in the local frame of the formation center, so one has
\begin{equation}
\left\{\begin{array}{ll}
\mathbf{r}_{i} &= \mathbf{r}_c+\mathbf{\bunderline{C}}_{BI}^T\left(\psi_c\right)\boldsymbol{p}_{ri}\\
\mathbf{\dot{r}}_{i} &= \mathbf{\dot{r}}_c+\mathbf{\dot{\bunderline{C}}}_{BI}^T\left(\psi_c\right)\boldsymbol{p}_{ri} = \mathbf{\dot{r}}_c+\mathbf{\bunderline{C}}_{BI}^T\left(\psi_c\right)\boldsymbol{\omega}_c^\times\boldsymbol{p}_{ri}\\
\mathbf{\ddot{r}}_{i} &= \mathbf{\ddot{r}}_c+\mathbf{\bunderline{C}}_{BI}^T\left(\psi_c\right)\boldsymbol{\omega}_c^\times\boldsymbol{\omega}_c^\times\boldsymbol{p}_{ri}+\mathbf{\bunderline{C}}_{BI}^T\left(\psi_c\right)\boldsymbol{\dot{\omega}}_c^\times\boldsymbol{p}_{ri}
\end{array}\right.
\label{Eq: VL_Sim_Refs}
\end{equation}
where $ \mathbf{r}_c=\left[x_c\text{, }y_c\text{, }z_c\right]^T$, $\mathbf{\bunderline{C}}_{BI}\left(\psi_c\right)$ is a rotation matrix and $\boldsymbol{\omega}_c^\times$ is a cross-product matrix, which are given as below, respectively.
\begin{equation*}
\mathbf{\bunderline{C}}_{BI}\left(\psi_c\right)=\left[\begin{array}{ccc}
\cos{\psi_c} & \sin{\psi_c} & 0 \\
-\sin{\psi_c} & \cos{\psi_c}& 0 \\
	0	    &		0	   & 1
\end{array}\right]\text{,} \qquad \boldsymbol{\omega}_c^\times=\left[\begin{array}{ccc}
	0 	& -\dot{\psi} c & 0 \\
\dot{\psi} c & 	0 	    &  0\\
	0	& 	0 	    &  0
\end{array}\right]
\end{equation*}
Although constant $\boldsymbol{p}_{ri}$ is assumed, the same method can be extended to the case with time-varying $\boldsymbol{p}_{ri} $.  Note that reference trajectories by (\ref{Eq: VL_Sim_Refs})  are not smooth, if there are any sudden changes in the accelerations of the formation center or angular rates of the local frame. Additionally, if a UAV is initially far away from its designated reference position, dramatical control efforts will be needed, which might violates control constraints. Hence, instead of implementing (\ref{Eq: VL_Sim_Refs}) directly, a cooperative filter is introduced as given below.
\begin{equation}
\left\{\begin{array}{ll}
\mathbf{\dot{\widehat{r}}}_{i}&=\mathbf{\widehat{v}}_{i} \\
\mathbf{\dot{\widehat{v}}}_{i} &=\mathbf{\ddot{r}}_{i}-\boldsymbol{\kappa}_{\mathbf{p}}\mathbf{\widehat{e}}_{\mathbf{p}i}-\boldsymbol{\kappa}_{\mathbf{v}}\mathbf{\widehat{e}}_{\mathbf{v}i}-\underset{j\in\mathscr{N}_{i}}{\sum}\left[\boldsymbol{c}_{\mathbf{p}}\left(\mathbf{\widehat{e}}_{\mathbf{p}i}-\mathbf{\widehat{e}}_{\mathbf{p}j}\right)+\boldsymbol{c}_{\mathbf{v}}\left(\mathbf{\widehat{e}}_{\mathbf {v}i}-\mathbf{\widehat{e}}_{\mathbf{v}j}\right)\right] \label{Eq: CoopFilter} 
\end{array}
\right.
\end{equation}
where $\mathbf{\widehat{r}}_{{i}}=\left[\widehat{x}_{i}\text{, }\widehat{y}_{i}\text{, }\widehat{z}_{i}\right]^T$ is the position vector of the $i$-th virtual leader in the inertial frame, $\mathbf{\widehat{v}}_{i}$ is the velocity vector, $\mathbf{\widehat{e}}_{\mathbf{p}i}=\mathbf{\widehat{r}}_{i}-\mathbf{{r}}_{{i}}$, $\mathbf{\widehat{e}}_{\mathbf{v}i}=\mathbf{\widehat{v}}_{i}-\mathbf{{v}}_{{i}}$, and $\boldsymbol{\kappa}_{\mathbf{p}}$, $\boldsymbol{\kappa}_{\mathbf{v}}$, $\boldsymbol{c}_{\mathbf{p}}$, and $\boldsymbol{c}_{\mathbf{v}}$ are diagonal positive definite gain matrices. 
\begin{remark}
From a motion planning perspective, the virtual structure is introduced to ensure rigid formation shape when reference motions are planned for all UAVs. The cooperative filter (\ref{Eq: CoopFilter}) is employed to smooth the planned motions in order to make them feasible and applicable.
\end{remark}
\begin{theorem} \label{Thm: CoopFilter}
Consider the cooperative filter (\ref{Eq: CoopFilter}) and the undirected graph $\mathscr{G}$. For any initial states $\mathbf{\widehat{r}}_{{i}}\left(t_0\right)$ and $\mathbf{\widehat{v}}_{{i}}\left(t_0\right)$, there exist  $\mathbf{\widehat{r}}_{{i}}\left(t\right)\to \mathbf{{r}}_{{i}}\left(t\right)$ and $\mathbf{\widehat{v}}_{{i}}\left(t\right)\to \mathbf{{v}}_{{i}}$ exponentially as $t\to\infty$, $\forall i\in\mathscr{V}$. 
\end{theorem}
\begin{proof}
The fact that $\mathbf{\widehat{r}}_{{i}}\left(t\right)\to \mathbf{{r}}_{{i}}\left(t\right)$ and $\mathbf{\widehat{v}}_{{i}}\left(t\right)\to \mathbf{{v}}_{{i}}$ exponentially as $t\to\infty$  implies $\lim_{t\to\infty}\mathbf{\widehat{e}}_{\mathbf{p}i}\left(t\right)=\mathbf{0}$ and $\lim_{t\to\infty}\mathbf{\widehat{e}}_{\mathbf{v}i}\left(t\right)=\mathbf{0}$, respectively. Define $\mathbf{\widehat{e}}_{\mathbf{p}}=\left[\mathbf{\widehat{e}}_{\mathbf{p}1}^T\text{, }\mathbf{\widehat{e}}_{\mathbf{p}2}^T\text{, }\ldots\text{, }\mathbf{\widehat{e}}_{\mathbf{p}n}^T\right]^T$ and $\mathbf{\widehat{e}}_{\mathbf{v}}=\left[\mathbf{\widehat{e}}_{\mathbf{v}1}^T\text{, }\mathbf{\widehat{e}}_{\mathbf{v}2}^T\text{, }\ldots\text{, }\mathbf{\widehat{e}}_{\mathbf{v}n}^T\right]^T$, so 
\begin{equation}
\mathbf{\dot{\widehat{e}}}_{\mathbf{v}}=-\left(\mathbf{I}_n\otimes\boldsymbol{\kappa}_{\mathbf{p}}+\boldsymbol{\mathcal{L}}\otimes\boldsymbol{c}_{\mathbf{p}}\right)\mathbf{\widehat{e}_p}-\left(\mathbf{I}_n\otimes\boldsymbol{\kappa}_{\mathbf{v}}+\boldsymbol{\mathcal{L}}\otimes\boldsymbol{c}_{\mathbf{v}}\right)\mathbf{\widehat{e}_v} \label{Eq: CFilter_Multi}
\end{equation}
The following Lyapunov function is chosen.
\begin{equation}
\mathbb{V}=\frac{1}{2}{\mathbf{\widehat{e}_p}^T\left(\mathbf{I}_n\otimes\boldsymbol{\kappa}_{\mathbf{p}}+\boldsymbol{\mathcal{L}}\otimes\boldsymbol{c}_{\mathbf{p}}\right)\mathbf{\widehat{e}_p}}+\frac{1}{2}{\mathbf{\widehat{e}_v}^T\mathbf{\widehat{e}_v}}
\end{equation}
In terms of Lemma \ref{Lem: LaplacianMat_Property}, $\mathbb{V}$ is positive definite. Differentiating $\mathbb{V}$ yields
\begin{eqnarray*}
\dot{\mathbb{V}} &=&{\mathbf{\widehat{e}_p}^T\left(\mathbf{I}_n\otimes\boldsymbol{\kappa}_{\mathbf{p}}+\boldsymbol{\mathcal{L}}\otimes\boldsymbol{c}_{\mathbf{p}}\right)
					\mathbf{\dot{\widehat{e}}_p}}+\mathbf{\widehat{e}_v}^T\mathbf{\dot{\widehat{e}}_v}\\
			    &=&\mathbf{\widehat{e}_p}^T\left(\mathbf{I}_n\otimes\boldsymbol{\kappa}_{\mathbf{p}}+\boldsymbol{\mathcal{L}}\otimes\boldsymbol{c}_{\mathbf{p}}\right)
			    		\mathbf{\widehat{e}_v} -\mathbf{\widehat{e}_v}^T\left(\mathbf{I}_n\otimes\boldsymbol{\kappa}_{\mathbf{p}}+\boldsymbol{\mathcal{L}}
					\otimes\boldsymbol{c}_{\mathbf{p}}\right)\mathbf{\widehat{e}}_{\mathbf{p}} \\
			    & &-\mathbf{\widehat{e}_v}^T\left(\mathbf{I}_n\otimes\boldsymbol{\kappa}_{\mathbf{v}}+\boldsymbol{\mathcal{L}}\otimes\boldsymbol{c}_{\mathbf{v}}\right)
					\mathbf{\widehat{e}}_{\mathbf{v}} \\
			    &=&-\mathbf{\widehat{e}_v}^T\left(\mathbf{I}_n\otimes\boldsymbol{\kappa}_{\mathbf{v}}+\boldsymbol{\mathcal{L}}\otimes\boldsymbol{c}_{\mathbf{v}}\right)\mathbf{\widehat{e}}_{\mathbf{v}} 
\end{eqnarray*}
Since both $\boldsymbol{\kappa}_{\mathbf{v}}$ and $\boldsymbol{c}_{\mathbf{v}}$ are positive definite, $\mathbf{I}_n\otimes\boldsymbol{\kappa}_{\mathbf{v}}+\boldsymbol{\mathcal{L}}\otimes\boldsymbol{c}_{\mathbf{v}}$ is positive definite, which implies that $\dot{\mathbb{V}}\leq 0$. Let $\mathbb{E}=\left\{\left(\mathbf{\widehat{e}_p}\text{, }\mathbf{\widehat{e}_v}\right)|\dot{\mathbb{V}}=0\right\}$. To ensure $\dot{\mathbb{V}}=0$, there must exist $\mathbf{\widehat{e}_v}\equiv \mathbf{0}$, which implies $\mathbf{\dot{\widehat{e}}_v}\equiv \mathbf{0}$. According to (\ref{Eq: CFilter_Multi}), $\mathbf{\dot{\widehat{e}}_v}\equiv \mathbf{0}$ and $\mathbf{\widehat{e}_v}\equiv \mathbf{0}$ will imply that $\mathbf{\widehat{e}}_{\mathbf{p}}\equiv \mathbf{0}$. Hence, the latest invariant set in $\mathbb{E}$ is $\mathbb{M}=\left\{\left(\mathbf{\widehat{e}_p}\text{, }\mathbf{\widehat{e}_v}\right)|\mathbf{\widehat{e}_p}=\mathbf{0}\text{, }\mathbf{\widehat{e}_v}=\mathbf{0}\right\}$. According LaSalle's Invariance principle  (c. f. Page 128 \cite{Khalil2002Book}), it follows that $\lim_{t\to\infty}\mathbf{\widehat{e}_p}\to\mathbf{0}$ and $\lim_{t\to\infty}\mathbf{\widehat{e}_v}\to\mathbf{0}$ asymptotically. As the cooperative filter (\ref{Eq: CoopFilter}) is a linear system, one is able to conclude that the error system (\ref{Eq: CFilter_Multi}) is exponentially stable. Hence, there exist  $\mathbf{\widehat{r}}_{{i}}\left(t\right)\to \mathbf{{r}}_{{i}}\left(t\right)$ and $\mathbf{\widehat{v}}_{{i}}\left(t\right)\to \mathbf{{v}}_{{i}}$ exponentially as $t\to\infty$, $\forall i\in\mathscr{V}$.
\end{proof}

\begin{figure}[bph]
\centering
\includegraphics[width=1\textwidth]{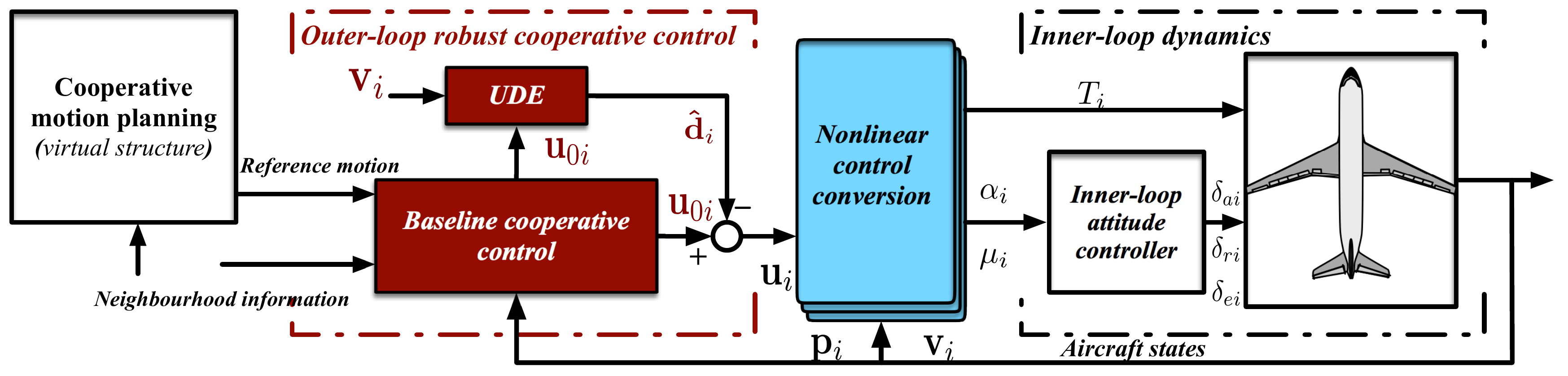}\\ \vspace{-3mm}
\caption{UDE-based robust cooperative formation control}
\label{Fig: RobCoopCntrlArchitect} \vspace{-5mm}
\end{figure}
 \section{Robust cooperative formation control} \label{Sec: CntrlDesign}
The general control law is expressed as 
\begin{equation}
u_{\rho i}=u_{\rho 0 i}-\widehat{d}_{\rho i} \text{, }\qquad \rho \in \left\{x\text{, }y\text{, }z\right\} \qquad  \left(i\in\mathscr{V}\right) \label{Eq: General_RobCoopCntrl}
\end{equation}
where $u_{\rho 0 i}$ is the baseline cooperative control law,  and $\widehat{d}_{\rho i}$ is the estimation of ${d}_{\rho i}$.  Shown in Figure \ref{Fig: RobCoopCntrlArchitect} is the  robust cooperative control block diagram.  
\subsection{Uncertainty and disturbance observer design}
Assume that  $u_{\rho i 0}$ has been well designed to stabilize the nominal double integrator system (\ref{Eq: DoubleIntegator}) when $d_{\rho i}=0$. Substituting (\ref{Eq: General_RobCoopCntrl}) into (\ref{Eq: DoubleIntegator}) yields
\begin{equation}
\ddot{\rho}_i=u_{\rho 0i}-\widehat{d}_{\rho i}+d_{\rho i}, \quad \rho\in\left\{x\text{, }y\text{, }z\right\} \label{Eq: DoubleInteg2}
\end{equation}
In terms of results in \cite{Zhong2004JDSMC, Zhu2015ACC, Zhu2018IJRNC},  a stable low-pass filter $G_{\rho i}(s)$ with unity gain is used to estimate $d_{\rho i}$, that is
\begin{equation}
G_{\rho i}(s)={1}\left/\left({\mathcal{T}_{\rho i}s+1}\right)\right., \quad \rho\in\left\{x\text{, }y\text{, }z\right\} \label{Eq: LowPassFilter}
\end{equation}
Hence,  
\begin{equation}
\widehat{d}_{\rho i}=\mathcal{L}^{-1}\left\{G_{\rho i}(s)\right\}*d_{\rho i}, \quad \rho\in\left\{x\text{, }y\text{, }z\right\}  \label{Eq: UDE1}
\end{equation}
where  $\mathcal{L}^{-1}\left\{\cdot\right\}$ is the inverse Laplace  transform.   Applying (\ref{Eq: UDE1}) to (\ref{Eq: DoubleInteg2})  yields
\begin{equation}
\ddot{\rho}_i=u_{\rho 0i}+\mathcal{L}^{-1}\left\{1-G_{\rho i}(s)\right\}*d_{\rho i}, \quad \rho\in\left\{x\text{, }y\text{, }z\right\} \label{Eq: UDE-Sys}
\end{equation}
where $1-G_{\rho i}(s)=\frac{\mathcal{T}_{\rho i}s}{\mathcal{T}_{\rho i}s+1}$ are high-pass filters. If the bandwidth of $G_{\rho i}(s)$ is chosen properly, the impact of $d_{\rho i}$ on the system could be ruled out.  In terms of  (\ref{Eq: DoubleInteg2}), we could denote $d_{\rho i}$ as a function of $\ddot{\rho}_i$, $u_{\rho 0i}$, and $\widehat{d}_{\rho i}$.  The following applicable uncertainty and disturbance observer is eventually obtained.
\begin{equation}
\widehat{d}_{\rho i}=\frac{1}{\mathcal{T}_{\rho i}}{\left(\dot{\rho}_i\left(t\right)-\dot{\rho}_i\left(0\right)-\int^t_0u_{\rho 0i}dt\right)},\quad \rho\in\left\{x\text{, }y\text{, }z\right\} \label{Eq: UDE2}
\end{equation}
Define $\widetilde{d}_{\rho i}=\widehat{d}_{\rho i}-d_{\rho i}$ as approximation errors, and the following lemma exists.
\begin{lemma}[\cite{Zhu2015CEP}]\label{Lem: UDE}
If the initial condition of (\ref{Eq: UDE2}) is chosen to be nill, namely $\widehat{d}_{\rho i}=0$ with $\rho\in\left\{x\text{, }y\text{, }z\right\}$,  the following conclusions exist.
\begin{enumerate}
\item $\widetilde{d}_{\rho i}$ are globally uniformly bounded by $\vert\widetilde{d}_{\rho i}\vert\leq \max\left\{\vert{{d}}_{\rho i}\left(0\right)\vert\text{, }\mathcal{T}_{\rho i}\Vert\dot{{d}}_{\rho i}\Vert_\infty\right\} $;
\item if $\vert{{d}}_{\rho i}\left(0\right)\vert\neq \mathcal{T}_{\rho i}\Vert\dot{{d}}_{\rho i}\Vert_\infty$, there exist any small positive constant $\epsilon$ and time $t_{\epsilon}$ such that $\vert\widetilde{d}_{\rho i}\vert<\mathcal{T}_{\rho i}\Vert\dot{{d}}_{\rho i}\Vert_\infty+\epsilon$ for all $t>t_{\epsilon}$, where $t_{\epsilon}=\max\left\{0\text{, } \mathcal{T}_{\rho i}\ln\left(\vert\frac{\vert{{d}}_{\rho i}\left(0\right)\vert-\mathcal{T}_{\rho i}\Vert\dot{{d}}_{\rho i}\Vert_\infty}{\epsilon}\vert\right)\right\}$;
\item if $\lim_{t\to\infty} \dot{{d}}_{\rho i} =0 $, $\lim_{t\to\infty} \widetilde{d}_{\rho i} =0 $.
\end{enumerate}
\end{lemma}

\subsection{Baseline cooperative control design}
When being rewritten in a matrix form, the control law (\ref{Eq: General_RobCoopCntrl}) is 
\begin{equation}
\mathbf{u}_{i}=\mathbf{u}_{0 i}-\mathbf{\widehat{d}}_{i} \qquad  \left(i\in\mathscr{V}\right)\label{Eq: General_RobCoopCntrl_Mat}
\end{equation}
Let $\mathbf{e}_{\mathbf{p}i}=\mathbf{p}_i-\mathbf{\tilde{r}}_i=\left[e_{xi}\text{, }e_{yi}\text{, }e_{zi}\right]^T$ and $\mathbf{e}_{\mathbf{v}i}=\mathbf{v}_i-\mathbf{\tilde{v}}_i=\left[e_{\dot{x}i}\text{, }e_{\dot{y}i}\text{, }e_{\dot{z}i}\right]^T$. The proposed baseline cooperative control for the $i$-th UAV  is
\begin{equation}
\mathbf{u}_{0i}=\mathbf{\ddot{r}}_i-\boldsymbol{K_{p}}\mathbf{e}_{\mathbf{p}i}-\boldsymbol{K_{v}}\mathbf{e}_{\mathbf{v}i}-\sum_{j\in\mathscr{N}_{i}}\left[\boldsymbol{C_{p}}\left(\mathbf{e}_{\mathbf{p}i}-\mathbf{e}_{\mathbf{p}j}\right)+\boldsymbol{C_{v}}\left(\mathbf{e}_{\mathbf {v}i}-\mathbf{e}_{\mathbf{v}j}\right)\right] \label{Eq: VL_CoopCntrl}
\end{equation}
where $\boldsymbol{K_{p}}\in\mathbb{R}^{3\times3}$, $\boldsymbol{K_{v}}\in\mathbb{R}^{3\times3}$, $\boldsymbol{C_{p}}\in\mathbb{R}^{3\times3}$, and $\boldsymbol{C_{v}}\in\mathbb{R}^{3\times3}$ are all positive definite diagonal parameter matrices given as below.
\begin{equation*}
\begin{array}{ll}
\boldsymbol{K_{p}}=diag\left\{K_{\boldsymbol{p}x}\text{, }K_{\boldsymbol{p}y}\text{, }K_{\boldsymbol{p}z}\right\}\text{, }\quad&\boldsymbol{K_{v}}=diag\left\{K_{\boldsymbol{v}x}\text{, }K_{\boldsymbol{v}y}\text{, }K_{\boldsymbol{v}z}\right\} \\
\boldsymbol{C_{p}}=diag\left\{C_{\boldsymbol{p}x}\text{, }C_{\boldsymbol{p}y}\text{, }C_{\boldsymbol{p}z}\right\}\text{, }\quad&\boldsymbol{C_{v}}=diag\left\{C_{\boldsymbol{v}x}\text{, }C_{\boldsymbol{v}y}\text{, }C_{\boldsymbol{v}z}\right\}
\end{array}
\end{equation*} 
Note that the same set of control gains are chosen, as we consider multiple homogeneous UAV models in close formation flight.  When substituting (\ref{Eq: VL_CoopCntrl}) and (\ref{Eq: General_RobCoopCntrl_Mat}) into (\ref{Eq: DoubleInteg_Mat}), a closed-loop error system is obtained.
\begin{eqnarray}
\mathbf{\dot{e}}_{\mathbf{v}i}&=&-\boldsymbol{K_{p}}\mathbf{e}_{\mathbf{p}i}-\boldsymbol{K_{v}}\mathbf{e}_{\mathbf{v}i}-\sum_{j\in\mathscr{N}_{i}}a_{ij}\left[\boldsymbol{C_{p}}\left(\mathbf{e}_{\mathbf{p}i}-\mathbf{e}_{\mathbf{p}j}\right)+\boldsymbol{C_{v}}\left(\mathbf{e}_{\mathbf {v}i}-\mathbf{e}_{\mathbf{v}j}\right)\right] \nonumber \\
&&+\boldsymbol{\kappa}_{\mathbf{p}}\mathbf{\widehat{e}}_{\mathbf{p}i}+\boldsymbol{\kappa}_{\mathbf{v}}\mathbf{\widehat{e}}_{\mathbf{v}i}+\underset{j\in\mathscr{N}_{i}}{\sum}\left[\boldsymbol{c}_{\mathbf{p}}\left(\mathbf{\widehat{e}}_{\mathbf{p}i}-\mathbf{\widehat{e}}_{\mathbf{p}j}\right)+\boldsymbol{c}_{\mathbf{v}}\left(\mathbf{\widehat{e}}_{\mathbf {v}i}-\mathbf{\widehat{e}}_{\mathbf{v}j}\right)\right] -\mathbf{\widetilde{d}}_i \label{Eq: VL_CLdyn_Single}
\end{eqnarray}
where $\mathbf{\widetilde{d}}_i=\left[\widetilde{d}_{xi}\text{, }\widetilde{d}_{yi}\text{, }\widetilde{d}_{zi}\right]^T$. Let $\mathbf{e_p}=\left[\mathbf{e}_{\mathbf{p}1}^T\text{, }\ldots\text{, }\mathbf{e}_{\mathbf{p}n}^T\right]^T$, $\mathbf{e_v}=\left[\mathbf{e}_{\mathbf{v}1}^T\text{, }\ldots\text{, }\mathbf{e}_{\mathbf{v}n}^T\right]^T$, and $\mathbf{\widetilde{d}}=\left[\mathbf{\widetilde{d}}_1^T\text{, }\ldots\text{, }\mathbf{\widetilde{d}}_n^T\right]^T$. The closed-loop tracking error dynamics are
\begin{eqnarray}
\mathbf{\dot{e}}_{\mathbf{v}}&=&-\left(\mathbf{I}_n\otimes\boldsymbol{K_{p}}+\boldsymbol{\mathcal{L}}\otimes\boldsymbol{C_{p}}\right)\mathbf{e_p}-\left(\mathbf{I}_n\otimes\boldsymbol{K_{v}}+\boldsymbol{\mathcal{L}}\otimes\boldsymbol{C_{v}}\right)\mathbf{e_v}\nonumber \\
&&+\left(\mathbf{I}_n\otimes\boldsymbol{\kappa}_{\mathbf{p}}+\boldsymbol{\mathcal{L}}\otimes\boldsymbol{c}_{\mathbf{p}}\right)\mathbf{\widehat{e}_p}+\left(\mathbf{I}_n\otimes\boldsymbol{\kappa}_{\mathbf{v}}+\boldsymbol{\mathcal{L}}\otimes\boldsymbol{c}_{\mathbf{v}}\right)\mathbf{\widehat{e}_v} 
-\mathbf{\widetilde{d}} \label{Eq: VL_CLdyn_Multi}
\end{eqnarray}
where $\mathbf{I}_n\in\mathbb{R}^{n\times n}$ is an identity matrix.

\section{Stability analysis} \label{Sec: StabAnal}
\begin{lemma}[Lemma 4.7 in \cite{Khalil2002Book}] \label{Lem: Cascaded-system}
Consider a cascade system
\begin{subequations}
\begin{align}
\mathbf{\dot{x}}_1&=\boldsymbol{f}_1\left(t\text{, }\mathbf{x}_1\text{, }\mathbf{x}_2\right) \label{Eq: Cascade1}\\
\mathbf{\dot{x}}_2&=\boldsymbol{f}_2\left(t\text{, }\mathbf{x}_2\right) \label{Eq: Cascade2}
\end{align}
\end{subequations}
where $\boldsymbol{f}_1:\left[0\text{, }\infty\right)\times\mathbb{R}^{n_{1}\times 1}\times\mathbb{R}^{n_{2}\times 1}\rightarrow\mathbb{R}^{n_{1}\times 1}$ and  $\boldsymbol{f}_2:\left[0\text{, }\infty\right)\times\mathbb{R}^{n_{2}\times 1}\rightarrow\mathbb{R}^{n_{2}\times 1}$ are piecewise continuous in $t$ and locally Lipschitz in $\left[\mathbf{x}_1^T\text{, }\mathbf{x}_2^T\right]^T$.  If the system (\ref{Eq: Cascade1}), with $\mathbf{x}_2$ as input, is input-to-state stable and the system (\ref{Eq: Cascade2}) is globally uniformly asymptotically stable, the entire cascade system (\ref{Eq: Cascade1}) and (\ref{Eq: Cascade2}) is globally uniformly asymptotically stable.
\end{lemma}
The stability analysis is divided into two steps. At the first step, it is shown that the nominal system of (\ref{Eq: VL_CLdyn_Multi}) without the consideration of any uncertainties and disturbances could be stabilized by the baseline cooperative controller (\ref{Eq: VL_CoopCntrl}). Once the nominal system is stabilized, the uncertainty and disturbance estimator design will be validated and the estimation errors will be ensured to be uniformly and ultimately bounded according to Lemma \ref{Lem: UDE} under Assumption \ref{Assump: UncerDist}. At the second step, the stability of closed-loop system (\ref{Eq: VL_CLdyn_Multi}) will be analyzed. For the first step, we assume $\mathbf{\widetilde{d}}_i=\mathbf{0}$. The nominal tracking error dynamics (\ref{Eq: VL_CLdyn_Multi}) is thus
\begin{equation}
\mathbf{\dot{e}}_{\mathbf{v}}=-\left(\mathbf{I}_n\otimes\boldsymbol{K_{p}}+\boldsymbol{\mathcal{L}}\otimes\boldsymbol{C_{p}}\right)\mathbf{e_p}-\left(\mathbf{I}_n\otimes\boldsymbol{K_{v}}+\boldsymbol{\mathcal{L}}\otimes\boldsymbol{C_{v}}\right)\mathbf{e_v}+\boldsymbol{\sigma}_{\mathbf{\widehat{e}}}\label{Eq: VL_CLdyn_Multi_Nom}
\end{equation}
where $\boldsymbol{\sigma}_{\mathbf{\widehat{e}}}=\left(\mathbf{I}_n\otimes\boldsymbol{\kappa}_{\mathbf{p}}+\boldsymbol{\mathcal{L}}\otimes\boldsymbol{c}_{\mathbf{p}}\right)\mathbf{\widehat{e}_p}+\left(\mathbf{I}_n\otimes\boldsymbol{\kappa}_{\mathbf{v}}+\boldsymbol{\mathcal{L}}\otimes\boldsymbol{c}_{\mathbf{v}}\right)\mathbf{\widehat{e}_v} $. According to Theorem \ref{Thm: CoopFilter}, $\mathbf{\widehat{e}_p}$ and $\mathbf{\widehat{e}_v}$ will converge to zero exponentially, so $\lim_{t\to\infty}\boldsymbol{\sigma}_{\mathbf{\widehat{e}}}\to\mathbf{0}$ exponentially. If the nominal system (\ref{Eq: VL_CLdyn_Multi}) is input-to-state stable with respect to $\boldsymbol{\sigma}_{\mathbf{\widehat{e}}}$, it would be asymptotically stable by Lemma \ref{Lem: Cascaded-system}.

The Laplacian matrix $\boldsymbol{\mathcal{L}}$ is a real symmetric matrix, so there exists an orthogonal constant matrix which diagonalizes $\boldsymbol{\mathcal{L}}$. Let $\mathbfcal{Q}=\left[\mathbfcal{Q}_1\text{, }\ldots\text{, }\mathbfcal{Q}_n\right]\in\mathbb{R}^{n\times n}$ be the orthogonal matrix that diagonalizes $\boldsymbol{\mathcal{L}}$, namely $\mathbfcal{Q}^T\boldsymbol{\mathcal{L}}\mathbfcal{Q}=\boldsymbol{\Lambda}$ where $\boldsymbol{\Lambda}=diag\left\{\lambda_1\text{, }\ldots\text{, }\lambda_n\right\}$. Choosing $\mathbfcal{Q}_1=\frac{1}{\sqrt{n}}\mathbf{1}_n$ yields $\lambda_1=0$ and $\boldsymbol{\Lambda}=diag\left\{0,\mathbfcal{\bar{Q}}^T\boldsymbol{\mathcal{L}}\mathbfcal{\bar{Q}}\right\}=diag\left\{0\text{, }\boldsymbol{\bar{\Lambda}}\right\}$, 
where $\mathbfcal{\bar{Q}}=\left[\mathbfcal{Q}_2\text{, }\ldots\text{, }\mathbfcal{Q}_n\right]\in\mathbb{R}^{n\times \left(n-1\right)}$ and $\boldsymbol{\bar{\Lambda}}=diag\left\{\lambda_2\text{, }\lambda_3\text{, }\ldots\text{, }\lambda_n\right\}>0$ with $\lambda_2$, $\lambda_3$, $\ldots$, $\lambda_n$ denoting positive eigenvalues of $\boldsymbol{\mathcal{L}}$. In terms of $\mathbfcal{Q}$, $\mathbf{e_p}$ and $\mathbf{e_v}$ will be transformed into a new coordinate system.  Let $\boldsymbol{\varepsilon_p}=\left(\mathbfcal{Q}^T\otimes \mathbf{I}_3\right)\mathbf{e_p}=\left[\boldsymbol{\varepsilon}_{\boldsymbol{p}1}^T\text{, }\boldsymbol{\varepsilon}_{\boldsymbol{p}2}^T\text{, }\ldots\text{, }\boldsymbol{\varepsilon}_{\boldsymbol{p}n}^T\right]^T$ and $\boldsymbol{\varepsilon_v}=\left(\mathbfcal{Q}^T\otimes \mathbf{I}_3\right)\mathbf{e_v}=\left[\boldsymbol{\varepsilon}_{\boldsymbol{v}1}^T\text{, }\boldsymbol{\varepsilon}_{\boldsymbol{v}2}^T\text{, }\ldots\text{, }\boldsymbol{\varepsilon}_{\boldsymbol{v}n}^T\right]^T$, where $\boldsymbol{\varepsilon}_{\boldsymbol{p}i}=\left[\varepsilon_{\boldsymbol{p}xi}\text{, }\varepsilon_{\boldsymbol{p}yi}\text{, }\varepsilon_{\boldsymbol{p}zi}\right]^T\in\mathbb{R}^{3\times1}$ and $\boldsymbol{\varepsilon}_{\boldsymbol{v}i}=\left[\varepsilon_{\boldsymbol{v}xi}\text{, }\varepsilon_{\boldsymbol{v}yi}\text{, }\varepsilon_{\boldsymbol{v}zi}\right]^T\in\mathbb{R}^{3\times1}$ with $i\in\mathscr{V}$.
Note that $\mathbfcal{Q}^T\otimes \mathbf{I}_3$ is invertible, as both $\mathbfcal{Q}$ and $\mathbf{I}_3$ are invertible. Hence,
\begin{equation}
\left\{\begin{array}{l}
\Vert\mathbf{e_p}\Vert_2\leq \Vert\left(\mathbfcal{Q}\otimes \mathbf{I}_3\right)\Vert_2 \Vert\boldsymbol{\varepsilon_p}\Vert_2 = \Vert\boldsymbol{\varepsilon_p}\Vert_2\\
\Vert\mathbf{e_v}\Vert_2\leq \Vert\left(\mathbfcal{Q}\otimes \mathbf{I}_3\right)\Vert_2 \Vert\boldsymbol{\varepsilon_v}\Vert_2 = \Vert\boldsymbol{\varepsilon_v}\Vert_2
\end{array}\right. \quad \forall t>0 \label{Eq: E_Epsilon_Ineq}
\end{equation}
If $\lim_{t\to\infty} \Vert\boldsymbol{\varepsilon_p}\Vert_2=0$ and$\lim_{t\to\infty} \Vert\boldsymbol{\varepsilon_v}\Vert_2=0$, the close formation tracking is, therefore, achieved according to Definition \ref{Def: Formation_VL}. On the other hand, If $\lim_{t\to\infty} \Vert\boldsymbol{\varepsilon_p}\Vert_2\leq\epsilon_p$ and$\lim_{t\to\infty} \Vert\boldsymbol{\varepsilon_v}\Vert_2\leq\epsilon_v$, the bounded close formation tracking will be obtained according to Definition \ref{Def: Practical_Formation_VL}. Based on the coordinate transformation by $\mathbfcal{Q}^T\otimes \mathbf{I}_3$ and (\ref{Eq: VL_CLdyn_Multi}), one could get  
\begin{equation}
\boldsymbol{\dot{\varepsilon}_{v}}	=-\left(\mathbf{I}_n\otimes\boldsymbol{K_{p}}+\boldsymbol{\Lambda}
										\otimes\boldsymbol{C_{p}}\right)\boldsymbol{\varepsilon_p}-\left(\mathbf{I}_n\otimes\boldsymbol{K_{v}}+\boldsymbol{\Lambda}
										\otimes\boldsymbol{C_{v}}\right)\boldsymbol{\varepsilon_v}+\boldsymbol{\bar{\sigma}}_{\mathbf{\widehat{e}}}
\label{Eq: VL_CLdyn_Multi_Nom_New}
\end{equation}
where $\boldsymbol{\bar{\sigma}}_{\mathbf{\widehat{e}}}=\left(\mathbfcal{Q}^T\otimes \mathbf{I}_3\right)\boldsymbol{{\sigma}}_{\mathbf{\widehat{e}}}=\left[\boldsymbol{\bar{\sigma}}_1^T\text{, }\ldots\text{,}\boldsymbol{\bar{\sigma}}_n^T\right]^T$ with $\boldsymbol{\bar{\sigma}}_i=\left[\bar{\sigma}_{xi}\text{, }\bar{\sigma}_{yi}\text{, }\bar{\sigma}_{zi}\right]^T$ for $i=1$, $2$, $\ldots$, $n$. Since $\mathbfcal{Q}$ is nonsingular and $\lim_{t\to\infty}\boldsymbol{{\sigma}}_{\mathbf{\widehat{e}}}\to\mathbf{0}$, one has $\lim_{t\to\infty}\boldsymbol{\bar{\sigma}}_{\mathbf{\widehat{e}}}\to\mathbf{0}$.
The dynamics of $\boldsymbol{\varepsilon}_{\boldsymbol{v}}$ is decoupled, as $\boldsymbol{K_{p}}$, $\boldsymbol{K_{v}}$, $\boldsymbol{K_{p}}$, $\boldsymbol{C_{v}}$, and $\boldsymbol{\Lambda}$ are all diagonal matrices.
Therefore, we have
\begin{equation}
\left[\begin{array}{c}
{\dot{\varepsilon}}_{\boldsymbol{p}\rho i}\\
{\dot{\varepsilon}}_{\boldsymbol{v}\rho i}
\end{array}
\right]=\underbrace{\left[\begin{array}{cc}
			0				&		1	\\
-K_{\boldsymbol{p}\rho}-\lambda_iC_{\boldsymbol{p}\rho}	&	-K_{\boldsymbol{v}\rho}-\lambda_iC_{\boldsymbol{v}\rho}
\end{array}
\right]}_{\mathbf{A}_{\rho i}}\left[\begin{array}{c}
{{\varepsilon}}_{\boldsymbol{p}\rho i}\\
{{\varepsilon}}_{\boldsymbol{v}\rho i}
\end{array}
\right]+\left[\begin{array}{c}
0 \\
\bar{\sigma}_{\rho i}
\end{array}\right] \label{Eq: VL_CLdyn_Single_Nom_New}
\end{equation}
where $\rho\in\left\{x\text{, }y\text{, }z\right\}$ and $ i\in\mathscr{V}$. 
Since $K_{\boldsymbol{p}\rho}$, $C_{\boldsymbol{p}\rho}$, $K_{\boldsymbol{v}\rho}$, $C_{\boldsymbol{v}\rho}>0$, and $\lambda_i\geq0$, the system matrices $\mathbf{A}_{\rho i}$ are Hurwitz ($\forall\rho\in\left\{x\text{, }y\text{, }z\right\}$ and $\forall i\in\mathscr{V}$), so the system (\ref{Eq: VL_CLdyn_Single_Nom_New}) is input-to-state stable with respect to $\bar{\sigma}_{\rho i}$ according to Lemma \ref{Lem: Cascaded-system}. With the consideration of  $\lim_{t\to\infty}\boldsymbol{\bar{\sigma}}_{\mathbf{\widehat{e}}}\to\mathbf{0}$ and Lemma \ref{Lem: Cascaded-system}, one is able to conclude that the nominal system (\ref{Eq: VL_CLdyn_Single_Nom_New}) is asymptotically stable. Accordingly, the nominal system (\ref{Eq: VL_CLdyn_Multi_Nom}) is asymptotically stable. 
\begin{theorem}\label{Thm: VL_ISS}
Suppose Assumption \ref{Assump: UncerDist} holds. The closed-loop close formation tracking error system (\ref{Eq: VL_CLdyn_Multi}) with the uncertainty and disturbance observer design (\ref{Eq: UDE2}) is input-to-state stable with respect to $\mathbf{\dot{d}}=\left[\mathbf{\dot{d}}_1^T\text{, }\ldots\text{, }\mathbf{\dot{d}}_n^T\right]^T$ where $\mathbf{\dot{d}}_i=\left[\dot{d}_{xi}\text{, }\dot{d}_{yi}\text{, }\dot{d}_{zi}\right]^T$. Furthermore, the bounded close formation flight tracking is achieved by the uncertainty and disturbance estimator (\ref{Eq: UDE2}) and the virtual leader-based cooperative baseline law (\ref{Eq: VL_CoopCntrl}).
\end{theorem}
\begin{proof}
Let $\mathbf{e}=\left[\mathbf{e_p}^T\text{, }\mathbf{e_v}^T\right]^T$. The tracking error system (\ref{Eq: VL_CLdyn_Multi}) could be rewritten as 
\begin{eqnarray}
\mathbf{\dot{e}} &=&\underbrace{\left[\begin{array}{cc}
			\mathbf{0} 														& \mathbf{I}_n\otimes \mathbf{I}_3\\
-\mathbf{I}_n\otimes\boldsymbol{K_{p}}-\boldsymbol{\mathcal{L}}\otimes\boldsymbol{C_{p}} & -\mathbf{I}_n\otimes\boldsymbol{K_{v}}-\boldsymbol{\mathcal{L}}\otimes\boldsymbol{C_{v}}
\end{array}\right]}_{\mathbf{A}}\mathbf{e} +\underbrace{\left[\begin{array}{cc}
\mathbf{0}  \\
-\mathbf{I}_n\otimes \mathbf{I}_3
\end{array}\right]}_{\mathbf{B}}\nonumber \\
&&\times\left( \mathbf{\widetilde{d}} -\boldsymbol{\sigma}_{\mathbf{\widehat{e}}}\right)\label{Eq: MultErr_VL}
\end{eqnarray}
where $\mathbf{A}$ is Hurwitz, as the nominal system is asymptotically stable. Let $\mathbf{\widehat{e}}=\left[\mathbf{\widehat{e}_p}^T\text{, }\mathbf{\widehat{e}_v}^T\right]^T$, so (\ref{Eq: CFilter_Multi}) could be rewritten as
\begin{equation}
\mathbf{\dot{\widehat{e}}}=\underbrace{\left[\begin{array}{cc}
			\mathbf{0} 														& \mathbf{I}_n\otimes \mathbf{I}_3\\
-\mathbf{I}_n\otimes\boldsymbol{\kappa}_{\mathbf{p}}-\boldsymbol{\mathcal{L}}\otimes\boldsymbol{c}_{\mathbf{p}} & -\mathbf{I}_n\otimes\boldsymbol{\kappa}_{\mathbf{v}}-\boldsymbol{\mathcal{L}}\otimes\boldsymbol{c}_{\mathbf{v}}
\end{array}\right]}_{\mathbf{A_{\widehat{e}}}}\mathbf{\widehat{e}} 
\end{equation}
where $\mathbf{A_{\widehat{e}}}$ is Hurwitz. It is easy to obtain that $\boldsymbol{\sigma}_{\mathbf{\widehat{e}}}=\mathbf{B}^T\mathbf{A_{\widehat{e}}}\mathbf{\widehat{e}}$. 
Define 
\begin{equation*}
\mathbf{A_d}=diag\left\{-{1}/{\mathcal{T}_{x1}}\text{, }-{1}/{\mathcal{T}_{y1}}\text{, }-{1}/{\mathcal{T}_{z1}}\text{, }\ldots\text{, }-{1}/{\mathcal{T}_{xn}}\text{, }-{1}/{\mathcal{T}_{yn}}\text{, }-{1}/{\mathcal{T}_{zn}}\right\}
\end{equation*} 
where $\mathbf{A_d}$ is Hurwitz. The following  compact error dynamic model could be obtained for the uncertainty and disturbance observers. 
\begin{equation}
 \mathbf{\dot{\widetilde{d}}}=\mathbf{A_d}\mathbf{\widetilde{d}}-\mathbf{\dot{d}} \label{Eq: UDE_error_Vector}
\end{equation}
Define $\mathbf{e_a}=\left[\mathbf{e}^T\text{, }\mathbf{\widehat{e}}^T\text{, }\mathbf{\widetilde{d}}^T\right]^T$. The composite error system by (\ref{Eq: MultErr_VL}) and (\ref{Eq: UDE_error_Vector}) is
\begin{equation}
\mathbf{\dot{e}_a}=\underbrace{\left[\begin{array}{ccc}
\mathbf{A} & -\mathbf{B}\mathbf{B}^T\mathbf{A_{\widehat{e}}}  & \mathbf{B} \\
\mathbf{0}  & 			\mathbf{A_{\widehat{e}}}			 	&\mathbf{0} \\
\mathbf{0}  &		 	\mathbf{0}  					 	& \mathbf{A_d}
\end{array}\right]}_{\mathbf{A_a}}\mathbf{e_a}+\underbrace{\left[\begin{array}{c}
\mathbf{0} \\
\mathbf{0} \\
-\mathbf{I}_n\otimes \mathbf{I}_3
\end{array}\right]}_{\mathbf{B_a}} \mathbf{\dot{d}} \label{Eq: UDE_error_Vector}
\end{equation}
where $\mathbf{A_a}$ is always Hurwitz, as  $\mathbf{A}$, $\mathbf{A_{\widehat{e}}}$, and $\mathbf{A_d}$ are all Hurwitz. If $\mathbf{\dot{d}}$ is taken as inputs, the system (\ref{Eq: UDE_error_Vector}) is input-to-state stable with respect to $\mathbf{\dot{d}}$, which implies that (\ref{Eq: VL_CLdyn_Multi}) is input-to-state stable with respect to $\mathbf{\dot{d}}$.

For the bounded close formation tracking flight, one needs to show Definition \ref{Def: Practical_Formation_VL} is satisfied. In light of (\ref{Eq: E_Epsilon_Ineq}), we instead show $\boldsymbol{\varepsilon_p}$ and $\boldsymbol{\varepsilon_v}$ are ultimately bounded with arbitrarily small ultimate boundaries. Using $\boldsymbol{\varepsilon_p}=\left(\mathbfcal{Q}^T\otimes \mathbf{I}_3\right)\mathbf{e_p}$, $\boldsymbol{\varepsilon_v}=\left(\mathbfcal{Q}^T\otimes \mathbf{I}_3\right)\mathbf{e_v}$, and (\ref{Eq: VL_CLdyn_Multi_Nom_New}), one has 
\begin{equation}
\boldsymbol{\dot{\varepsilon}_{v}}	=-\left(\mathbf{I}_n\otimes\boldsymbol{K_{p}}+\boldsymbol{\Lambda}
										\otimes\boldsymbol{C_{p}}\right)\boldsymbol{\varepsilon_p}-\left(\mathbf{I}_n\otimes\boldsymbol{K_{v}}+\boldsymbol{\Lambda}
										\otimes\boldsymbol{C_{v}}\right)\boldsymbol{\varepsilon_v}-\mathbf{\widetilde{d}}_{\boldsymbol{\varepsilon}}+\boldsymbol{\bar{\sigma}}_{\mathbf{\widehat{e}}} \label{Eq: VL_CLdyn_Multi_New}
\end{equation}
where $\mathbf{\widetilde{d}}_{\boldsymbol{\varepsilon}}=\left(\mathbfcal{Q}^T\otimes \mathbf{I}_3\right)\mathbf{\widetilde{d}}=\left[\mathbf{\widetilde{d}}_{\boldsymbol{\varepsilon}1}^T\text{, }\ldots\text{, }\mathbf{\widetilde{d}}_{\boldsymbol{\varepsilon}n}^T\right]^T$. Since $\mathbfcal{Q}$ is an orthogonal matrix, we have $\Vert\mathbf{\widetilde{d}}_{\boldsymbol{\varepsilon}}\Vert_2\leq \Vert \left(\mathbfcal{Q}^T\otimes \mathbf{I}_3\right)\Vert_2\Vert\mathbf{\widetilde{d}}\Vert_2=\Vert\mathbf{\widetilde{d}}\Vert_2$.  If Assumption \ref{Assump: UncerDist} holds, $\mathbf{\widetilde{d}}$ is uniformly globally bounded according to Lemma \ref{Lem: UDE}, which implies $\mathbf{\widetilde{d}}_{\boldsymbol{\varepsilon}}$ is uniformly globally bounded.  In terms of (\ref{Eq: VL_CLdyn_Single_Nom_New}) for the nominal case, we have 
\begin{equation*}
\left[\begin{array}{c}
{\dot{\varepsilon}}_{\boldsymbol{p}\rho i}\\
{\dot{\varepsilon}}_{\boldsymbol{v}\rho i}
\end{array}
\right]=\underbrace{\left[\begin{array}{cc}
			0				&		1	\\
-K_{\boldsymbol{p}\rho}-\lambda_iC_{\boldsymbol{p}\rho}	&	-K_{\boldsymbol{v}\rho}-\lambda_iC_{\boldsymbol{v}\rho}
\end{array}
\right]}_{\mathbf{A}_{\rho i}}\left[\begin{array}{c}
{{\varepsilon}}_{\boldsymbol{p}\rho i}\\
{{\varepsilon}}_{\boldsymbol{v}\rho i}
\end{array}
\right]+\underbrace{\left[\begin{array}{c}
0\\
-1
\end{array}
\right]}_{\mathbf{B}_{\rho i}}\left(\widetilde{d}_{\boldsymbol{\varepsilon}\rho i}-\bar{\sigma}_{\rho i}\right)\label{Eq: VL_CLdyn_Single_Nom_New}
\end{equation*}
where $\rho\in\left\{x\text{, }y\text{, }z\right\}\text{, } i\in\mathscr{V}$.
Since $\mathbf{A}_{\rho i}$ is Hurwitz, the system (\ref{Eq: VL_CLdyn_Single_Nom_New}) is input-to-state stable with respect to $\widetilde{d}_{\boldsymbol{\varepsilon}\rho i}-\bar{\sigma}_{\rho i}$ for all $\rho\in\left\{x\text{, }y\text{, }z\right\}\text{, } i\in\mathscr{V}$. Define $\boldsymbol{\varepsilon}_{\rho i}=\left[{{\varepsilon}}_{\boldsymbol{p}\rho i}\text{, }{{\varepsilon}}_{\boldsymbol{v}\rho i}\right]^T=\boldsymbol{\varepsilon}_{\rho i}^{1}+\boldsymbol{\varepsilon}_{\rho i}^{2}$ where
\begin{equation*}
\left\{\begin{array}{ll}
\boldsymbol{\dot{\varepsilon}}_{\rho i}^{1} & =\mathbf{A}_{\rho i}\boldsymbol{\varepsilon}_{\rho i}^{1}- \mathbf{B}_{\rho i}\bar{\sigma}_{\rho i}\\
\boldsymbol{\dot{\varepsilon}}_{\rho i}^{2} & =\mathbf{A}_{\rho i}\boldsymbol{\varepsilon}_{\rho i}^{2}+\mathbf{B}_{\rho i}\widetilde{d}_{\boldsymbol{\varepsilon}\rho i}
\end{array}
\right.
\end{equation*}
As $\lim_{t\to\infty}\bar{\sigma}_{\rho i}\to\mathbf{0}$, one has $\lim_{t\to\infty}\boldsymbol{\varepsilon}_{\rho i}^{1}\to\mathbf{0}$ for any initial conditions. For $\boldsymbol{\varepsilon}_{\rho i}^{2}$,  the following inequality always exists
\begin{eqnarray*}
\left\Vert\boldsymbol{\varepsilon}_{\rho i}^{2}\right\Vert_2&\leq& \sqrt{\frac{\lambda_{max}\left(\mathbf{P}_{\rho i}\right)}{\lambda_{min}\left(\mathbf{P}_{\rho i}\right)}}e^{-\frac{t-t_{0}}{2\lambda_{max}\left(\mathbf{P}_{\rho i}\right)}}\left\Vert\boldsymbol{\varepsilon}_{\rho i}^{2}\left(t_0\right)\right\Vert_2+\left(1-e^{-\frac{t-t_0}{2\lambda_{max}\left(\mathbf{P}_{\rho i}\right)}}\right)\\
&&\times\frac{2\lambda_{max}^2\left(\mathbf{P}_{\rho i}\right)}{\lambda_{min}\left(\mathbf{P}_{\rho i}\right)}\Vert\widetilde{d}_{\boldsymbol{\varepsilon}\rho i}\Vert_{\mathscr{L}_{\infty}} 
\end{eqnarray*}
where $\mathbf{P}_{\rho i}>0$ such that $\mathbf{P}_{\rho i}\mathbf{A}_{\rho i}+\mathbf{A}_{\rho i}^T\mathbf{P}_{\rho i}=-\mathbf{I}_2$. Hence, 
\begin{equation*}
\lim_{t\to\infty} \left\Vert\boldsymbol{\varepsilon}_{\rho i}\right\Vert_2 \leq \lim_{t\to\infty}\left\Vert\boldsymbol{\varepsilon}_{\rho i}^{1}\right\Vert_2 +\lim_{t\to\infty}\left\Vert\boldsymbol{\varepsilon}_{\rho i}^{2}\right\Vert_2\leq \frac{2\lambda_{max}^2\left(\mathbf{P}_{\rho i}\right)}{\lambda_{min}\left(\mathbf{P}_{\rho i}\right)}\Vert\widetilde{d}_{\boldsymbol{\varepsilon}\rho i}\Vert_{\mathscr{L}_{\infty}} 
\end{equation*}
Therefore, $\boldsymbol{\varepsilon}_{\rho i}$ will be ultimately bounded for all $\rho\in\left\{x\text{, }y\text{, }z\right\}\text{, } i\in\mathscr{V}$, which implies $\mathbf{e_p}$ and $\mathbf{e_v}$ are both ultimately bounded. According to the definitions of $\mathbf{\widehat{e}}_{\mathbf{p}i}$, $\mathbf{\widehat{e}}_{\mathbf{v}i}$, $\mathbf{e}_{\mathbf{p}i}$, and $\mathbf{e}_{\mathbf{v}i}$, we have 
\begin{equation}
\mathbf{p}_i-\mathbf{r}_i=\mathbf{\widehat{e}}_{\mathbf{p}i}+\mathbf{e}_{\mathbf{p}i}\quad \mathbf{v}_i-\mathbf{\dot{r}}_i=\mathbf{\widehat{e}}_{\mathbf{v}i}+\mathbf{e}_{\mathbf{v}i}
\end{equation}
Since $\lim_{t\to\infty}\mathbf{\widehat{e}}_{\mathbf{p}i}\to\mathbf{0}$, $\lim_{t\to\infty}\mathbf{\widehat{e}}_{\mathbf{v}i}\to\mathbf{0}$, and $\mathbf{e}_{\mathbf{p}i}$ and $\mathbf{e}_{\mathbf{v}i}$ are both ultimately bounded,  there exist $\epsilon_p>0$ and $\epsilon_v>0$ such that
\begin{equation}
\lim_{t\to\infty} \Vert\mathbf{p}_i-\mathbf{r}_i\Vert_2 \leq\epsilon_p \qquad \text{and}\qquad\lim_{t\to\infty} \Vert\mathbf{v}_i-\mathbf{\dot{r}}_i\Vert_2 \leq \epsilon_v
\end{equation}
In addition, the transfer matrix from $\widetilde{d}_{\boldsymbol{\varepsilon}\rho i}$ to $\boldsymbol{{\varepsilon}}_{\boldsymbol{p}\rho i}$ and $\boldsymbol{{\varepsilon}}_{\boldsymbol{v}\rho i}$ are 
\begin{equation}
\mathbf{G}_{\rho i}\left(s\right)=\frac{1}{s^2+\left(K_{\boldsymbol{p}\rho}+\lambda_iC_{\boldsymbol{p}\rho}\right)s+\left(K_{\boldsymbol{v}\rho}+\lambda_iC_{\boldsymbol{v}\rho}\right)}\left[\begin{array}{c}
1 \\
s
\end{array}
\right] \label{Eq: TF_G_pd}
\end{equation}
By increasing $K_{\boldsymbol{v}\rho}+\lambda_iC_{\boldsymbol{v}\rho}$, the ultimate boundaries for $\boldsymbol{{\varepsilon}}_{\boldsymbol{p}\rho i}$ and $\boldsymbol{{\varepsilon}}_{\boldsymbol{v}\rho i}$ could be reduced according (c.f. \cite{Khalil2002Book}, Page 613). In sum, the bounded close formation flight tracking will be achieved by the proposed control law.
\end{proof}
 In general, a cooperative  formation control has better performance than a leader-follower formation controller. In the closed-loop error dynamics (\ref{Eq: VL_CLdyn_Single_Nom_New}), $K_{\boldsymbol{p}\rho}$ and $K_{\boldsymbol{v}\rho}$ are parameters of the trajectory tracking controller, while $\lambda_iC_{\boldsymbol{p}\rho}$ and $\lambda_iC_{\boldsymbol{v}\rho}$ are gains of the cooperative mechanism. The control parameters $K_{\boldsymbol{p}\rho}$ and $K_{\boldsymbol{v}\rho}$, which allow a UAV to track its reference trajectory, have the same roles to gains in a leader-follower controller. If $K_{\boldsymbol{p}\rho}$ and $K_{\boldsymbol{v}\rho}$ are kept to be constant, the introduction of the cooperative mechanism will potentially increase the gains of the closed-loop error dynamics through $\lambda_iC_{\boldsymbol{p}\rho}$ and $\lambda_iC_{\boldsymbol{v}\rho}$ as shown in (\ref{Eq: VL_CLdyn_Single_Nom_New}). Therefore, a cooperative formation controller could result in much faster responses than a leader-follower controller. As shown in  (\ref{Eq: TF_G_pd}), the existence of the cooperative mechanism could reduce the steady state gain from $\widetilde{d}_{\boldsymbol{\varepsilon}\rho i}$ to $\boldsymbol{{\varepsilon}}_{\boldsymbol{p}\rho i}$. Therefore, a cooperative controller has much smaller position tracking errors than a leader-follower controller, so more accurate close formation control could be achieved. In addition, due to the introduction of virtual leaders, a leader UAV will have much less influence on the tracking performance of a follower UAV in close formation.
\begin{theorem}\label{Thm: VL_Asymp}
asymptotic close formation tracking is achieved by the robust cooperative control (\ref{Eq: General_RobCoopCntrl_Mat}) with the uncertainty and disturbance observer given in (\ref{Eq: UDE2}) and the baseline cooperative control given in (\ref{Eq: VL_CoopCntrl}), if and only if $\lim_{t\to\infty}\dot{d}_{\rho i}=0$ with $\rho\in\left\{x\text{, }y\text{, }z\right\}$ and $i\in\mathscr{V}$.
\end{theorem}
\begin{proof}
\emph{Necessity:} The system matrix  $\mathbf{A}$ of the closed-loop error system (\ref{Eq: MultErr_VL}) is Hurwitz, so the system  (\ref{Eq: MultErr_VL}) is input-to-state stable with respect to $\mathbf{\widetilde{d}}$. Therefore, one must have $\lim_{t\to\infty} \Vert\mathbf{\widetilde{d}}\Vert_2 =0$ to guarantee $\lim_{t\to\infty} \Vert\mathbf{p}_i-\mathbf{r}_i\Vert_2 =0$ and $\lim_{t\to\infty} \Vert\mathbf{v}_i-\mathbf{\dot{r}}_i\Vert_2 =0$, namely $\lim_{t\to\infty} \Vert\mathbf{e}\Vert_2 =0$. According to (\ref{Eq: UDE_error_Vector}), $\lim_{t\to\infty} \Vert\mathbf{\dot{d}}\Vert_2 =0$ is a necessary condition for $\lim_{t\to\infty} \Vert\mathbf{\widetilde{d}}\Vert_2 =0$. Hence, to ensure  $\lim_{t\to\infty} \Vert\mathbf{p}_i-\mathbf{r}_i\Vert_2 =0$ and $\lim_{t\to\infty} \Vert\mathbf{v}_i-\mathbf{\dot{r}}_i\Vert_2 =0$, the necessary condition is  $\lim_{t\to\infty} \Vert\mathbf{\dot{d}}\Vert_2 =0$, namely $\lim_{t\to\infty}\dot{d}_{\rho i}$, $\forall \rho\in\left\{x\text{, }y\text{, }z\right\}$ and $\forall  i\in\mathscr{V}$.

\emph{Sufficiency:}
If $\lim_{t\to\infty}\dot{d}_{\rho i}$ with $\rho\in\left\{x\text{, }y\text{, }z\right\}$ and $i\in\mathscr{V}$, it is readily obtained that $\lim_{t\to\infty} \Vert\mathbf{\widetilde{d}}\Vert_2 =0$. Since $\mathbf{A}$ is Hurwitz, $\lim_{t\to\infty} \Vert\mathbf{e}\Vert_2 =0$ if $\lim_{t\to\infty} \Vert\mathbf{\widetilde{d}}\Vert_2 =0$. According to Lemma \ref{Lem: Cascaded-system},  $\lim_{t\to\infty}\dot{d}_{\rho i}=0$ with $\rho\in\left\{x\text{, }y\text{, }z\right\}$ and $i\in\mathscr{V}$ is sufficient enough to ensure $\lim_{t\to\infty} \Vert\mathbf{e}\Vert_2 =0$, namely $\lim_{t\to\infty} \Vert\mathbf{p}_i-\mathbf{r}_i\Vert_2 =0$ and $\lim_{t\to\infty} \Vert\mathbf{v}_i-\mathbf{\dot{r}}_i\Vert_2 =0$.

Therefore, $\lim_{t\to\infty}\dot{d}_{\rho i}=0$ with $\rho\in\left\{x\text{, }y\text{, }z\right\}$ and $i\in\mathscr{V}$ is a necessary and sufficient condition to gaurantee close formation flight tracking.
\end{proof}

\begin{table}
\centering \caption{UAV parameters} \label{TAB: F16Para} \footnotesize
\begin{tabular}{ccccc}
\toprule
\toprule
Parameter       &    Wing area  ($m^2$) 	&     Wing span   ($m$)& Mass 	($kg$)  & Drag coefficient       	\\ \toprule
 Value		&         27.87    &	    9.144		  &    9295.44 		&      0.0794 \\
\toprule
\toprule
\end{tabular}\vspace{-5mm}
\end{table} 
\section{Numerical simulations} \label{Sec: NumSim}
This section presents numerical simulation results which demonstrate the efficiency of the proposed virtual leader-based robust cooperative formation controller for close formation flight.  A group of virtual leaders are introduced, where each virtual leader will provide reference signals for a corresponding UAV in close formation as shown in Figure \ref{Fig: VL_Topology}. In the numerical simulations, the close formation flight problem of five UAVs is considered, namely $\mathscr{V}=\left\{1\text{, }2\text{, }3\text{, }4\text{, }5\right\}$. Necessary UAV parameters are given in Table \ref{TAB: F16Para}.  The formation aerodynamic disturbances are assumed to be unknown and generated using the aerodynamic model presented by \cite{Zhang2017JA}. 

According to the aerodynamic analysis in \cite{Zhang2017JA}), the optimal formation shape for close formation flight of five UAVs is given in Figure \ref{Fig: VL_Shape+Topology}.a. All UAVs are required to fly at the same altitude, and the optimal horizontal relative positions between two UAVs are defined in the body frame of the first virtual leader as shown in Figure \ref{Fig: VL_Shape+Topology}.a. The communication topology is illustrated in Figure \ref{Fig: VL_Shape+Topology}.b.   For any UAV $i$ and $j$ with $i$, $j\in\mathscr{V}$ and $i\neq j$, if there is a connection between them in Figure \ref{Fig: VL_Shape+Topology}.b, it implies that aicraft $i$ and $j$ are able to communicate with each other, and meanwhile $a_{ij}=1$, otherwise, $a_{ij}=0$. The adjacency matrix and degree matrix of the communication topology shown in Figure \ref{Fig: VL_Shape+Topology}.b are
\begin{equation*}
\Scale[0.85]{
\boldsymbol{\mathcal{A}}=\left[\begin{array}{ccccc}
0 & 1 & 1 & 0 & 0 \\
1 & 0 & 1 & 1 & 0 \\
1 & 1 & 0 & 1 & 1 \\
0 & 1 & 1 & 0 & 1 \\
0 & 0 & 1 & 1 & 0
\end{array}\right] \quad\quad\quad
\boldsymbol{\mathcal{D}}=\left[\begin{array}{ccccc}
2 & 0 & 0 & 0 & 0 \\
0 & 3 & 0 & 0 & 0 \\
0 & 0 & 4 & 0 & 0 \\
0 & 0 & 0 & 3 & 0 \\
0 & 0 & 0 & 0 & 2 
\end{array}\right]}
\end{equation*}\vspace{-7mm}
\begin{figure}[tbph]
\subfigure[The optimal formation shape]{\includegraphics[width=0.5\textwidth]{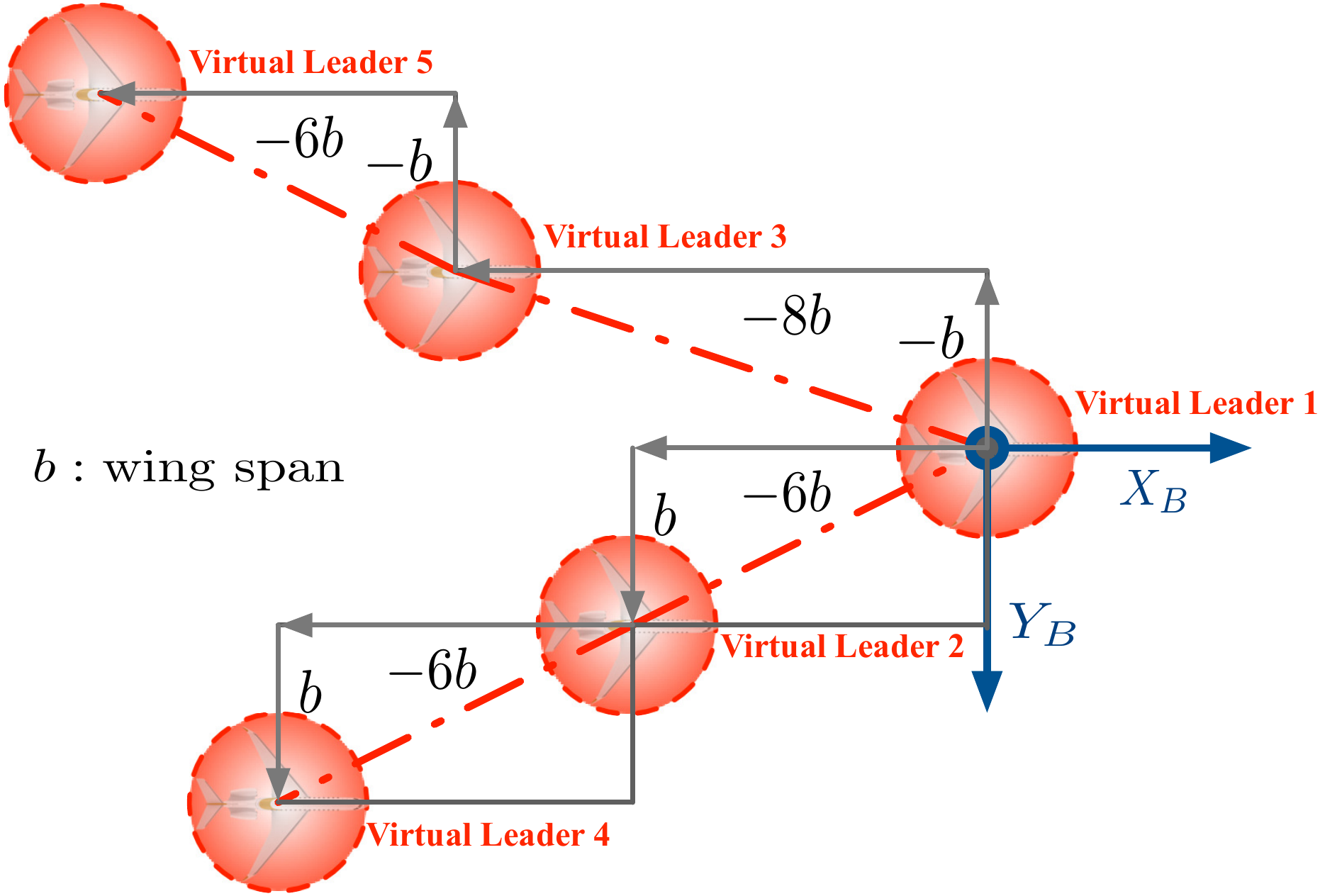}}\hfill
\subfigure[Communication topology]{\includegraphics[width=0.35\textwidth]{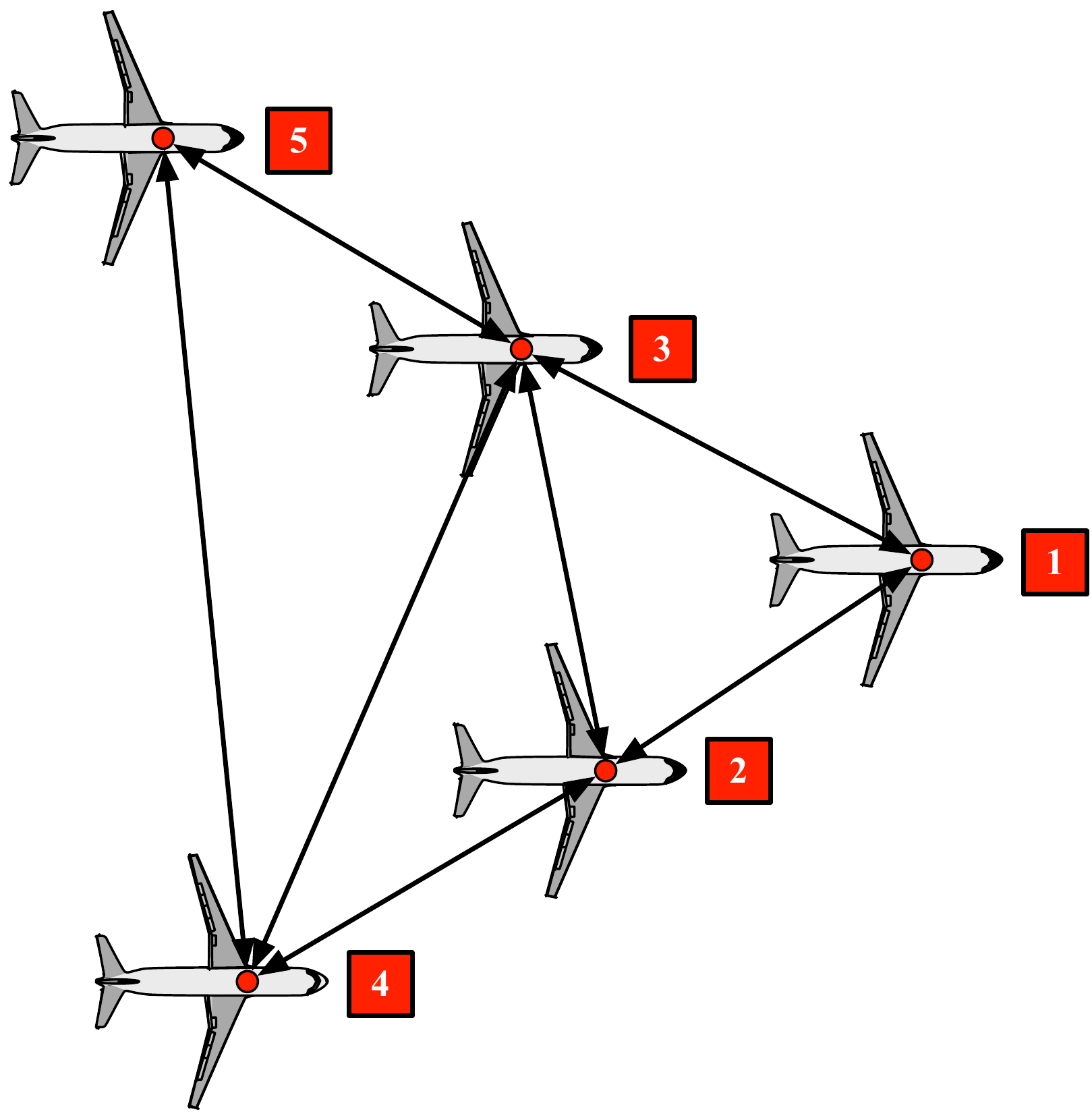}}\\ \vspace{-5mm}
\caption{ Virtual leader-based cooperative control}\vspace{-4mm}
\label{Fig: VL_Shape+Topology}
\end{figure}

If we treat the optimal formation shape in Figure \ref{Fig: VL_Shape+Topology}.a as a rigid body, the motion of each virtual leader could in general be resolved into a rotational motion around the formation geometric center and a translational motion which is equal to the translational motion of the formation center as shown in Figure \ref{Fig: VL_Motions}.a. The formation trajectory is introduced on the formation geometric center and described using the following navigation model.
\begin{equation}
\left\{\begin{array}{lll}
\dot{x}_{c}=V_{c} \cos{\gamma_{c}}\cos{\psi_{c}} &\dot{y}_{c}=V_{c} \cos{\gamma_{c}}\sin{\psi_{c}} & \dot{z}_{c}=-V_{c} \sin{\gamma_{c}}\\
\dot{V}_{c}= a_{Vc} &
\dot{\gamma}_{c}= a_{\gamma c} &
\dot{\psi}_{c}=a_{\psi c} 
\end{array}\right.\label{Eq: SysDyn_VL_Center}
\end{equation}
where $x_{c}$, $y_{c}$, and $z_{c}$ represent the position coordinates of the optimal formation center in the inertial frame, $V_{c}$, $\gamma_{c}$, and $\psi_{c}$ denote the ground velocity, flight path angle, and heading angle, respectively, and $a_{V c}$, $a_{\gamma c}$, and $a_{\psi c}$ specify the acceleration, flight path angular rate, and heading angular rate, respectively. 
The navigation model (\ref{Eq: SysDyn_VL_Center}) describes the motion of the geometric centre of the optimal formation shape shown in Figure \ref{Fig: VL_Shape+Topology}.a. The motion of all virtual leaders are calculated based on the motion of the geometric centre given in (\ref{Eq: SysDyn_VL_Center}). Let $\boldsymbol{p}_{ri}$ be the position vector of the $i$-th virtual leader in the body frame of the navigation model of the geometric centre as shown in Figure \ref{Fig: VL_Motions}.b.  According to Figure \ref{Fig: VL_Shape+Topology}.a, we have  $\boldsymbol{p}_{r1}=\left[8b\text{, }0\text{, }0\right]^T$, $\boldsymbol{p}_{r2}=\left[2b\text{, }b\text{, }0\right]^T$, $\boldsymbol{p}_{r3}=\left[0\text{, }-b\text{, }0\right]^T$, $\boldsymbol{p}_{r4}=\left[-4b\text{, }2b\text{, }0\right]^T$, and $\boldsymbol{p}_{r5}=\left[-6b\text{, }-b\text{, }0\right]^T$.
where $\boldsymbol{p}_{ri}$ are all constant for $i=1$, $\ldots$, $5$.

 Initially, $x_c\left(0\right)=26.87$ $m$, $y_c\left(0\right)=200$ $m$, $z_c\left(0\right)=-5000$ $m$, $V_c\left(0\right)=120$ $m/s$, and $\gamma_c\left(0\right)=\psi_c\left(0\right)=0$ $rad$. The acceleration is zero, namely $a_{Vc} =0 $, while angular rate signals are 
\begin{equation}
a_{\gamma c} =\left\{\begin{array}{rr}
\frac{\pi}{60}   & 10<t\leq 45 \\
-\frac{\pi}{60}  & 45<t\leq 80 \\
0			& otherwise
\end{array}\right.\text{, }\quad a_{\psi c} =\left\{\begin{array}{rr}
\frac{\pi}{1080}   & 10<t\leq 40 \\
-\frac{\pi}{1080}  & 50<t\leq 80 \\
0			& otherwise
\end{array}\right. \label{Eq: VL_Sim_Acc}
\end{equation}
Cooperative filter gains  are  $\boldsymbol{\kappa}_{\mathbf{p}}=diag\left\{1\text{, }1\text{, }1\right\}$, $\boldsymbol{\kappa}_{\mathbf{v}}=diag\left\{2.5\text{, }2.5\text{, }2.5\right\}$, $\boldsymbol{c}_{\mathbf{v}}=diag\left\{0.5\text{, }0.5\text{, }0.5\right\}$, and $\boldsymbol{c}_{\mathbf{p}}=diag\left\{1.25\text{, }1.25\text{, }1.25\right\}$. 
\begin{table}[tbp]
\centering \caption{Initial conditions} \label{TAB: Initial_VL} \footnotesize
\begin{tabular}{ccccccc}
\toprule
\toprule
    UAV    	&   \multicolumn{3}{c}{Position ($m$)} &     $V_T$ ($m/s$) 	&  $\gamma$ ($rad$)     	& $\psi$  ($rad$)    \\ 
      $\#$		&	$x$      & 	    $y$    & 	$z$ 	      &	   			&	     			      	&             		  \\ \toprule
	1		&   	190	    &	   190     &    -5005	      &           121		& 			0		&    		$0$		  \\
	2		&   	155	    &	   215     &    -5015	      &           116		& 			0		&    		$0$		  \\
	3		&   	140	    &	   182     &    -5005	      &           115		& 			0		&    $\frac{\pi}{120}$		  \\
	4		&   	85	    &	   225     &    -5015	      &           119		& 			0		&    		0		  \\
	5		&   	65	    &	   172     &    -5015	      &           120		& 			0		&    $\frac{\pi}{100}$		  \\
\toprule
\toprule
\end{tabular} \vspace{-5mm}
\end{table}

The initial conditions of the five UAVs are listed in Table \ref{TAB: Initial_VL}, while the same group of initial conditions are chosen for the cooperative filters. At the beginning all UAV is at level and straight flight with the original thrust $T_0=16954$ N. The same time constants are chosen for the uncertainty and disturbance observers of all UAVs, namely $\mathcal{T}_{\rho i}=0.2$ s with $\rho\in\left\{x\text{, }y\text{, }z\right\}$ and $i\in\mathscr{V}$. The baseline control parameters are given in Table \ref{TAB: CntrlPara_VL}. The trajectory tracking responses of the five UAVs are shown in Figure \ref{Fig: VL_TrajResp}, while the UAV trajectories at four crucial time periods are highlighted in Figure \ref{Fig: VL_Traj_TimeP}. Position tracking error responses are illustrated in Figures \ref{Fig: VL_LongP}, \ref{Fig: VL_LatP} and \ref{Fig: VL_VertP}, respectively, while the velocity tracking errors are presented in Figures \ref{Fig: VL_LongV}, \ref{Fig: VL_LatV} and \ref{Fig: VL_VertV}, respectively. The corresponding control inputs are given in Figures \ref{Fig: VL_Thrust}, \ref{Fig: VL_Lift} and \ref{Fig: VL_Mu}. Obviously, close formation flight tracking has been achieved by the proposed virtual-leader based robust cooperative control method.
\begin{table}[tbph]
\vspace{-6mm}
\centering \caption{Control parameters} \label{TAB: CntrlPara_VL} \footnotesize
\begin{tabular}{lccccccc}
\toprule
\toprule
   Parameters	&& $K_{\boldsymbol{p}x}$ & $K_{\boldsymbol{p}y}$ & 	$K_{\boldsymbol{p}z}$ & $K_{\boldsymbol{v}x}$	& $K_{\boldsymbol{v}y}$ & $K_{\boldsymbol{v}z}$  		  \\ \toprule
	Value       &&  	0.25	    	 &  		0.4	    	 &  		0.3	    	 &  		1.5	    	 &  		1.75	    	 &  		1.75	    	 	 \\ \toprule
\toprule
	Parameters & &$C_{\boldsymbol{p}x}$ & $C_{\boldsymbol{p}y}$ & $C_{\boldsymbol{p}z}$ & $C_{\boldsymbol{v}x}$ & $C_{\boldsymbol{v}y}$ & $C_{\boldsymbol{v}z}$ \\ \toprule
	Value       & &		0.15		 & 		0.15		& 			0.15		& 	0.55		 & 		0.55		 & 		0.55	  \\
\toprule
\toprule
\end{tabular}
\end{table}\vspace{-8mm}
\begin{figure}[h]
\centering
 \includegraphics[width=0.85\textwidth]{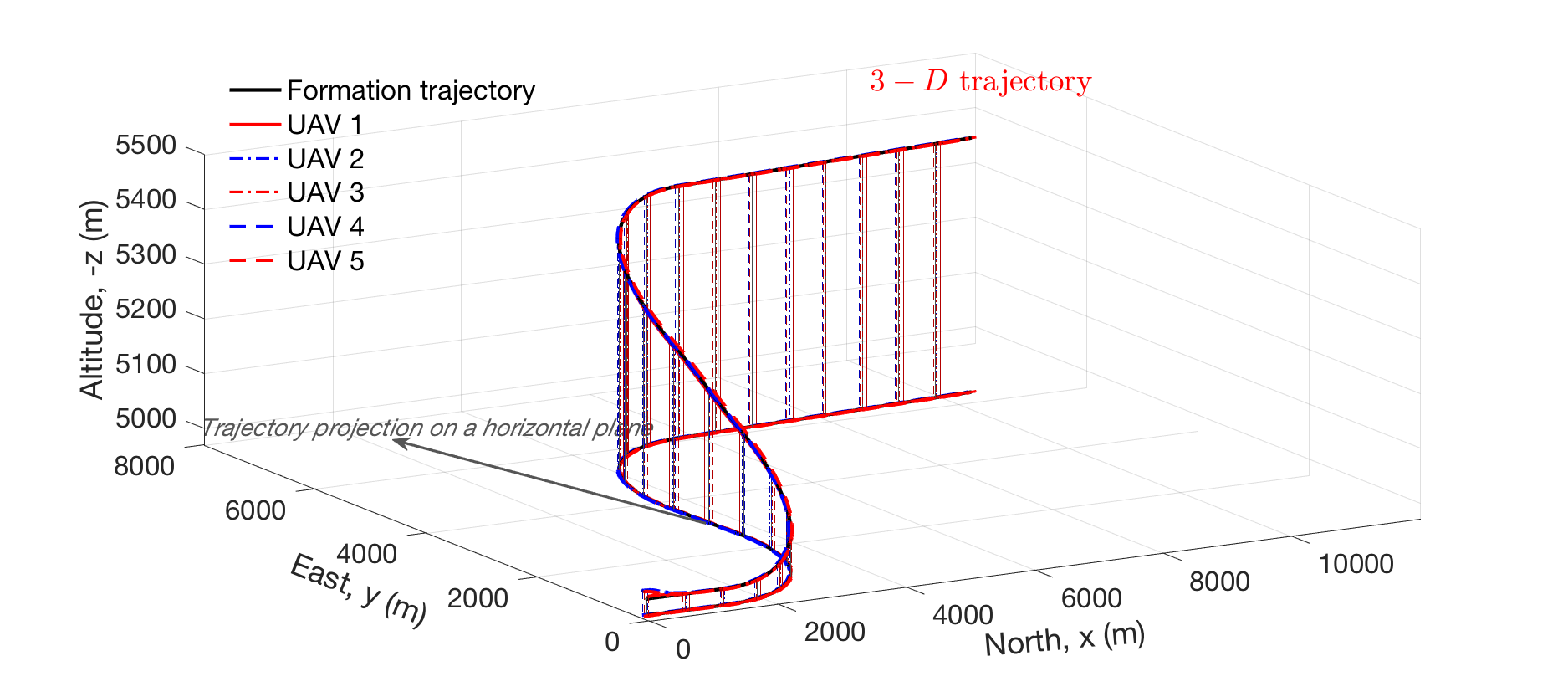} \\ \vspace{-5mm}
\caption{Trajectories of all five UAVs}\vspace{-5mm}
\label{Fig: VL_TrajResp}
\end{figure} 
\begin{figure}[h!]
\centering
\begin{tabular}{cc}
 \includegraphics[width=0.495\textwidth]{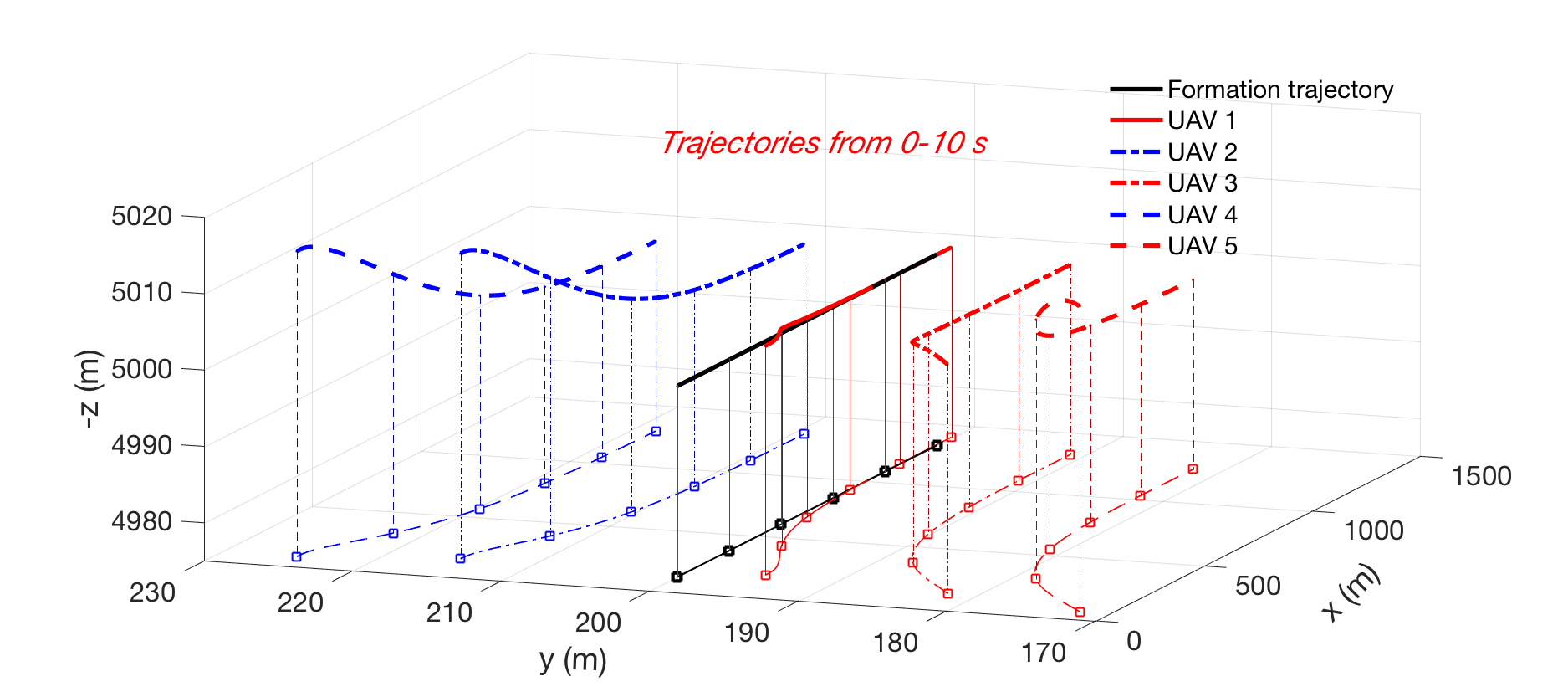}  \vspace{-1mm}&
  \includegraphics[width=0.495\textwidth]{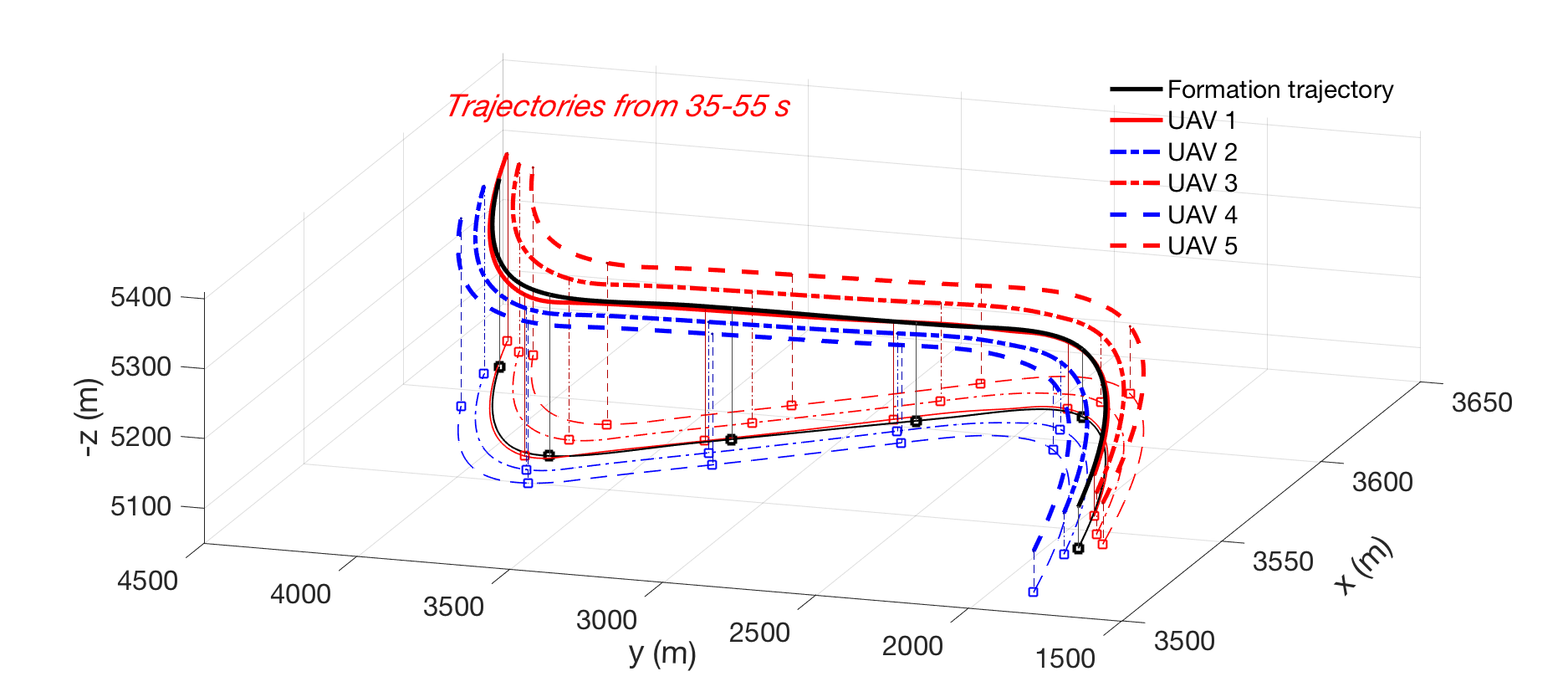}  \vspace{-1mm}\\ 
    \includegraphics[width=0.495\textwidth]{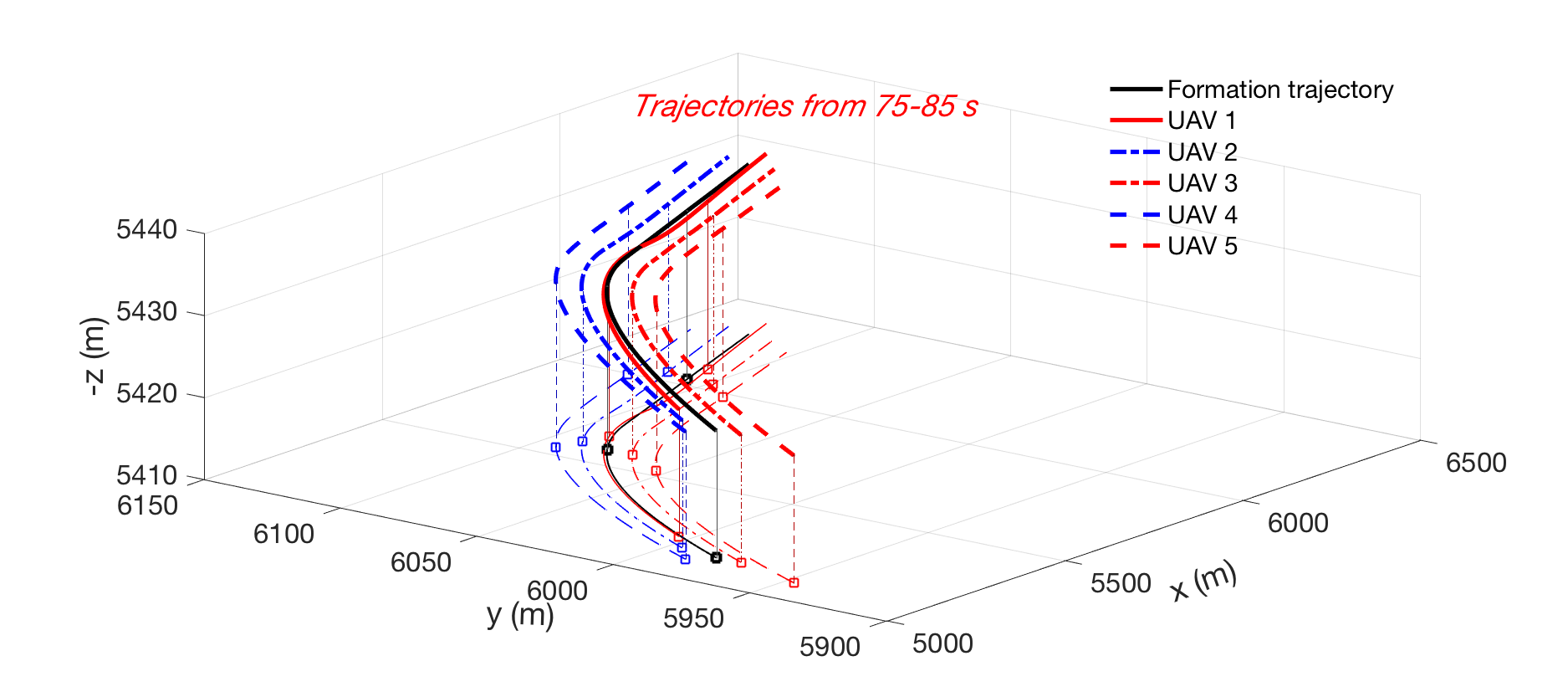}& 
      \includegraphics[width=0.495\textwidth]{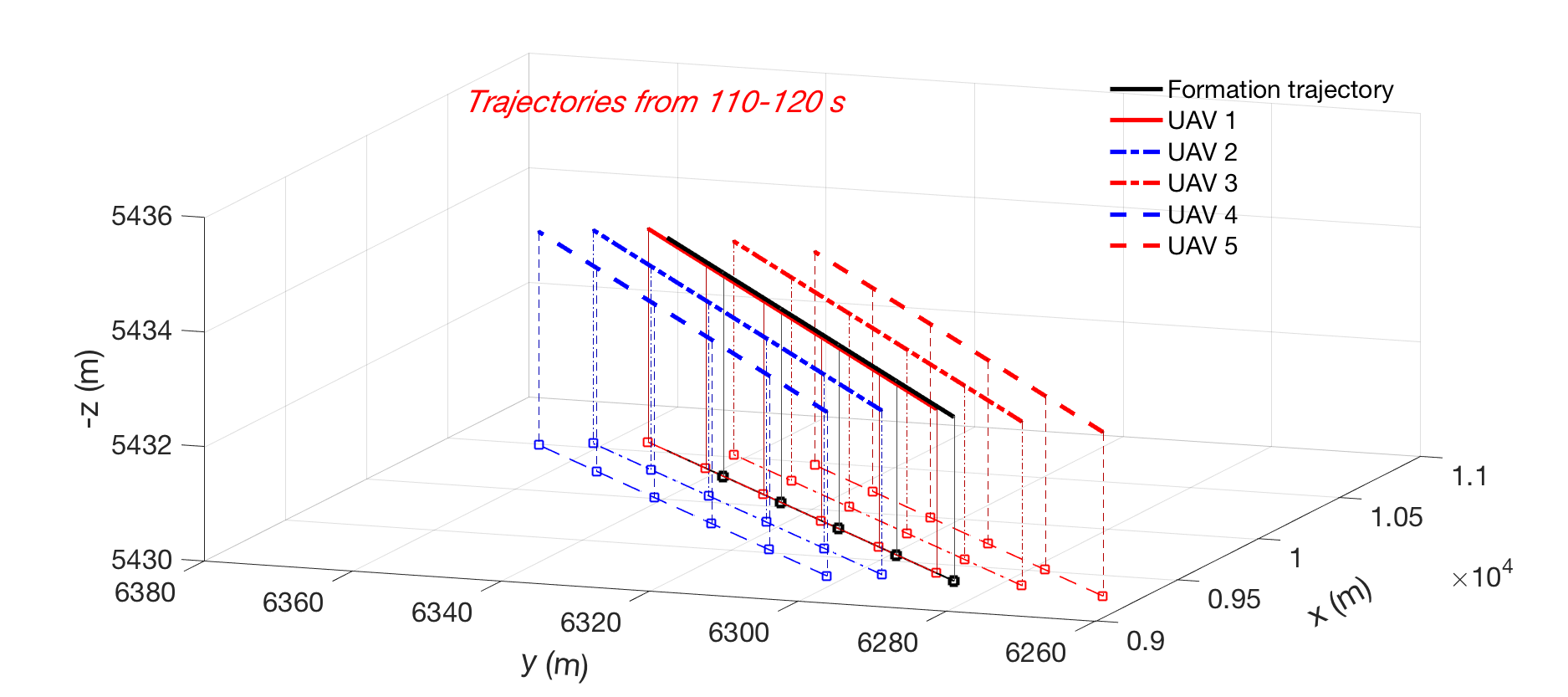}  \\
 \end{tabular}  \vspace{-5mm}
\caption{UAV trajectories at different time periods}  \vspace{-3mm}
\label{Fig: VL_Traj_TimeP}
\end{figure}

As indicated in Figures \ref{Fig: VL_LongV}-\ref{Fig: VL_VertV}, oscillations are observed when $t=10$ s, $40$ s, and $80$ s. It is because the reference trajectory signals are not smooth enough. According to (\ref{Eq: VL_Sim_Acc}), the angular rate signals $a_{\gamma c}$ and $a_{\chi c}$ have sudden changes at  $t=10$ s, $40$ s, and $80$ s.  The real acceleration of each virtual leader should be $\mathbf{\ddot{r}}_i = \mathbf{\ddot{r}}_c+\mathbf{\bunderline{C}}_{WI}^T\left(\psi_c\right)\boldsymbol{\omega}_c^\times\boldsymbol{\omega}_c^\times\boldsymbol{p}_{ri}+\mathbf{\bunderline{C}}_{WI}^T\left(\psi_c\right)\boldsymbol{\dot{\Omega}}_c^\times\boldsymbol{\boldsymbol{p}}_{ri}$ where the third term $\mathbf{\bunderline{C}}_{WI}^T\left(\psi_c\right)\boldsymbol{\dot{\Omega}}_c^\times\boldsymbol{\boldsymbol{p}}_{ri}$ represents the impact of changes in the angular rates of the reference trajectories. However,  $a_{\gamma c}$ and $a_{\chi c}$ are not changed smoothly, so $\mathbf{\bunderline{C}}_{WI}^T\left(\psi_c\right)\boldsymbol{\dot{\Omega}}_c^\times\boldsymbol{p}_{ri}$ will be infinity at $t=10$ s, $40$ s, and $80$ s. In order to avoid this infinity issue, the acceleration of each virtual leader is estimated by $\mathbf{\ddot{r}}_i = \mathbf{\ddot{r}}_c+\mathbf{\bunderline{C}}_{WI}^T\left(\psi_c\right)\boldsymbol{\omega}_c^\times\boldsymbol{\omega}_c^\times\boldsymbol{p}_{ri}$. Therefore, accelerations of all virtual leaders are underestimated at $t=10$ s, $40$ s, and $80$ s, which results in the oscillations in the responses. 

When in close formation flight, the very first UAV (UAV 1 in the simulations) is not flying  at any trailing vortices of other UAV, so it is the only UAV which doesn't receive any aerodynamic benefits. The first UAV has the same performance as a UAV at solo flight with the same conditions. Therefore, it could be used as an benchmark to show the benefits of close formation flight. In comparison with UAV 1, other UAV could experience at least $5.5\%$ decrease in their thrust inputs in close formation flight as shown in Figure \ref{Fig: VL_Thrust}. Hence, close formation flight could help follower UAV to reduce their drag and eventually help to save energies.
\begin{figure}[H]
\centering
 \includegraphics[width=1\textwidth]{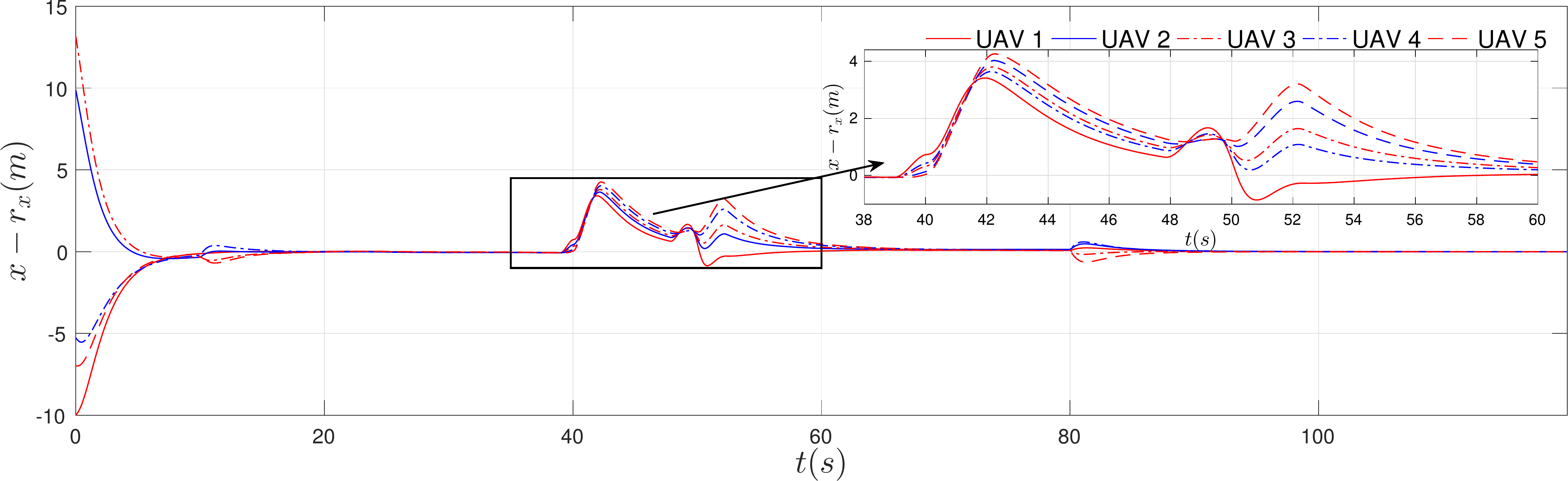} \vspace{-8mm}
\caption{Longitudinal position tracking errors}
\label{Fig: VL_LongP}\vspace{-5mm}
\end{figure}
\begin{figure}[H]
\centering
 \includegraphics[width=1\textwidth]{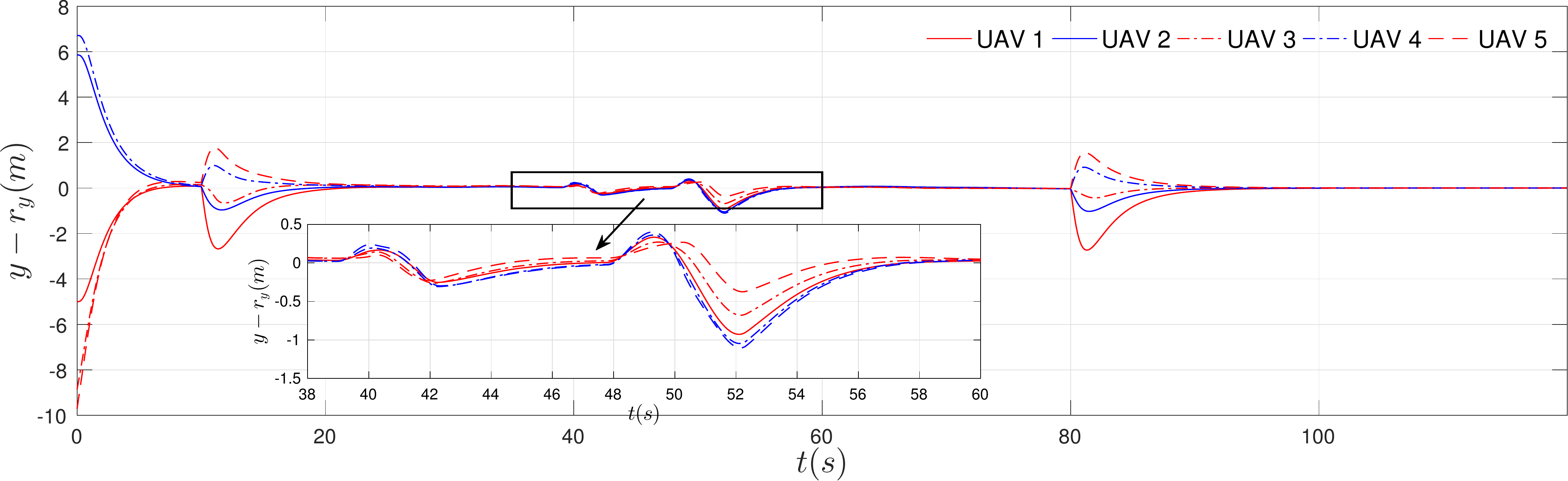} \vspace{-8mm}
\caption{Lateral position tracking errors}
\label{Fig: VL_LatP}\vspace{-5mm}
\end{figure}
\begin{figure}[H]
\centering
 \includegraphics[width=1\textwidth]{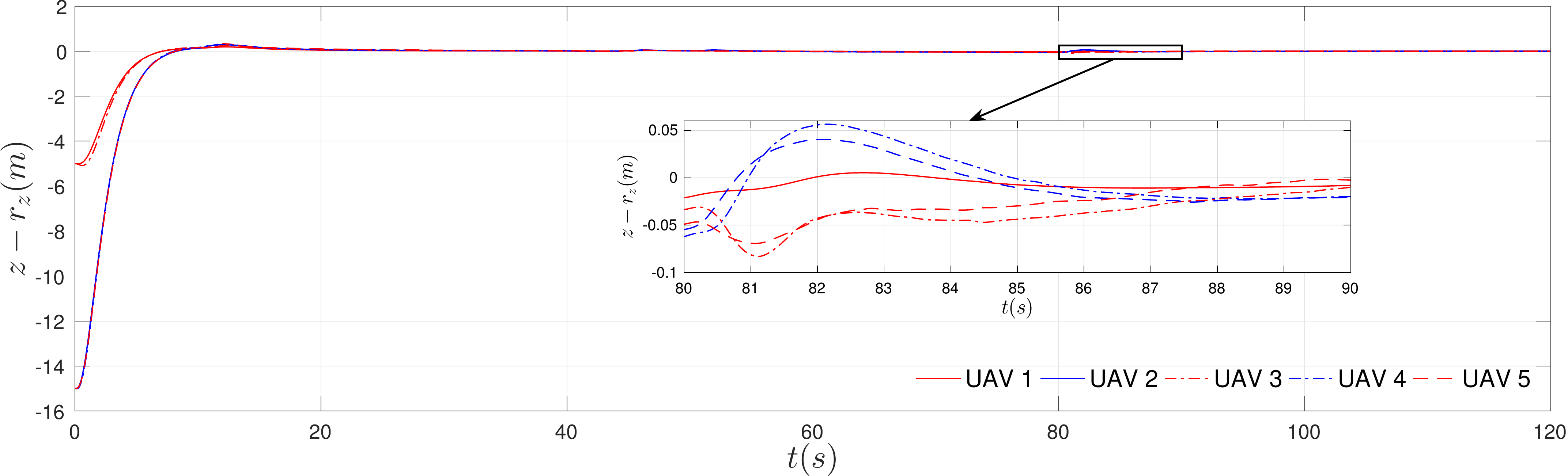} \vspace{-8mm}
\caption{Vertical position tracking errors}
\label{Fig: VL_VertP}
\end{figure}
 
\begin{figure}[H]
\begin{minipage}{0.46\textwidth}
\centering
 \includegraphics[width=1\textwidth]{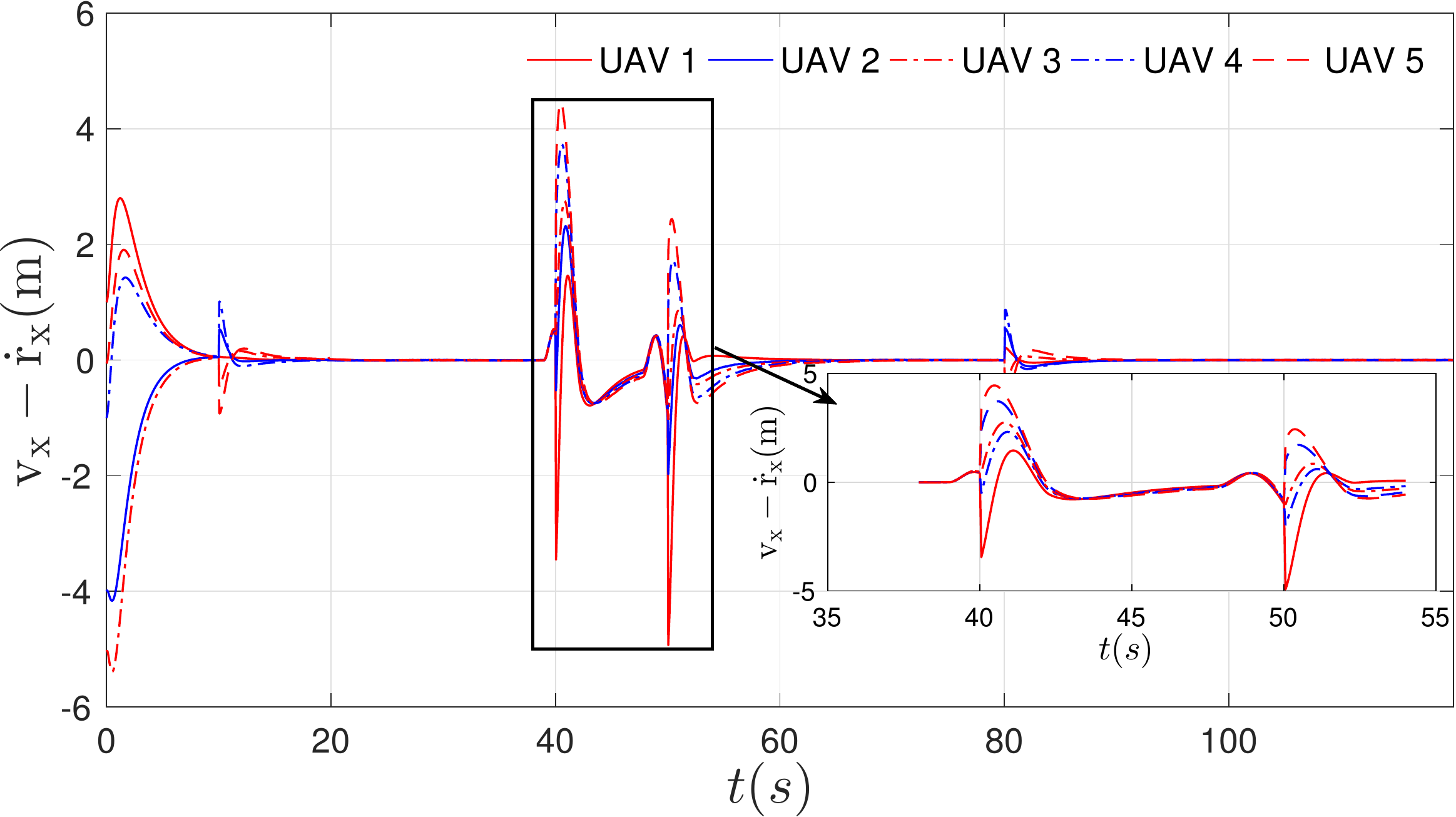} \\ \vspace{-3mm}
\caption{Velocity errors in $X_I$ axis}
\label{Fig: VL_LongV}
\end{minipage}\hfill
\begin{minipage}{0.46\textwidth}
\centering
\includegraphics[width=1\linewidth]{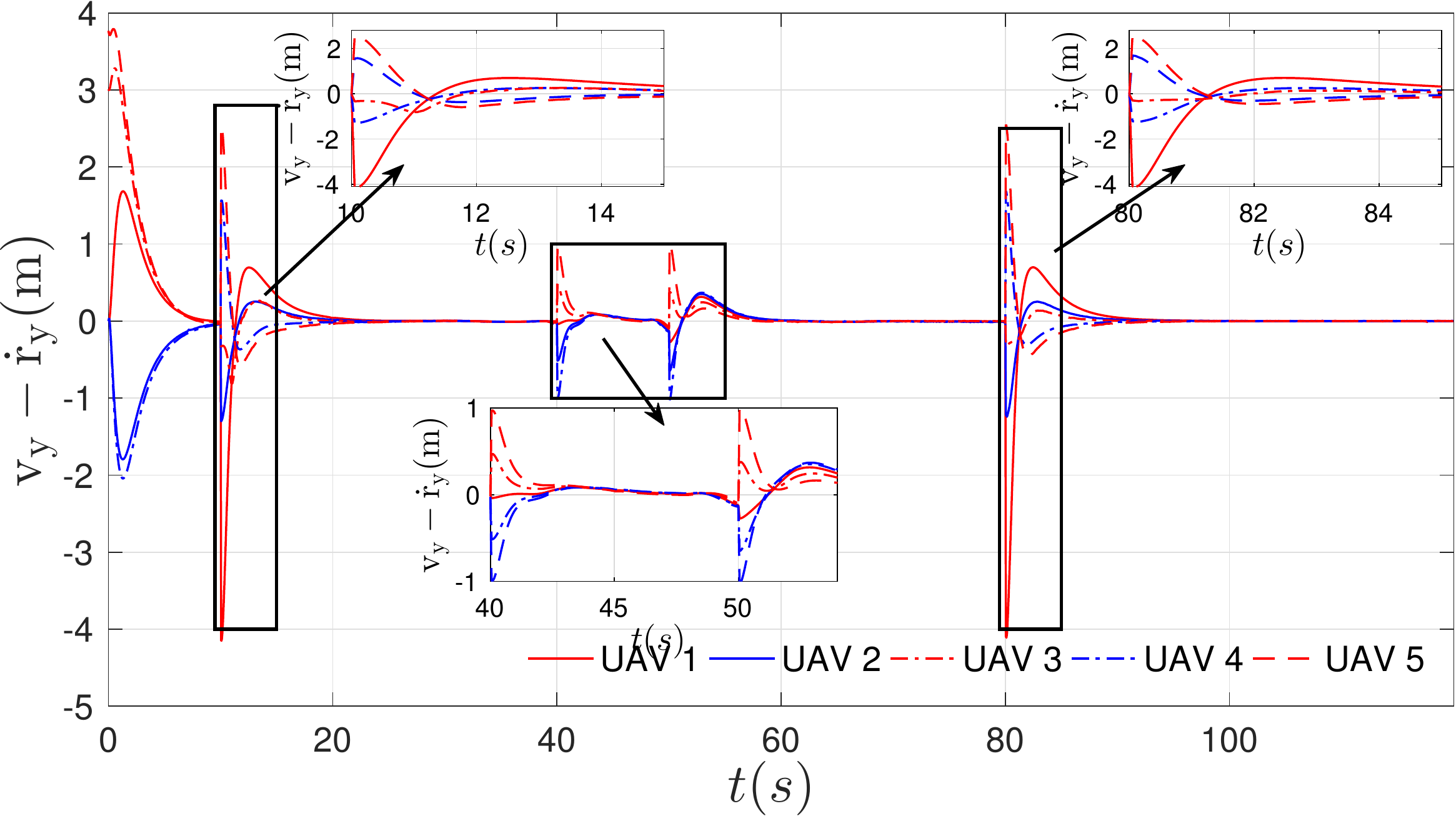}  \\ \vspace{-3mm}
\caption{Velocity errors in $Y_I$ axis}
\label{Fig: VL_LatV} 
\end{minipage}\\ \vspace{-5mm}
\end{figure}
\begin{figure}[H]
\centering
\begin{minipage}{0.46\textwidth}
\centering
 \includegraphics[width=1\textwidth]{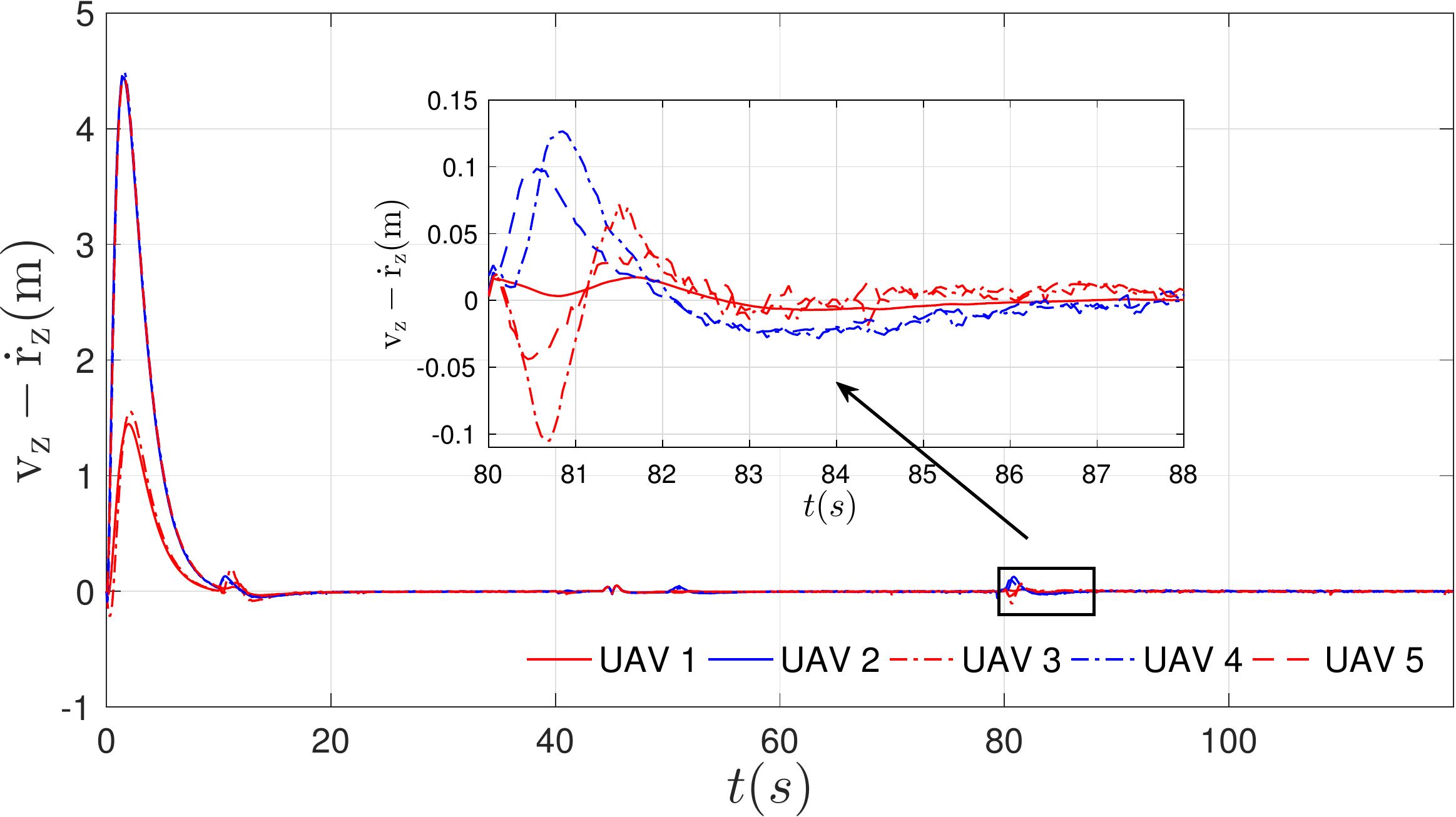} \\ \vspace{-3mm}
\caption{Velocity errors in $Z_I$ axis}
\label{Fig: VL_VertV}
\end{minipage}\hfill
\begin{minipage}{0.46\textwidth}
\centering
\includegraphics[width=1\linewidth]{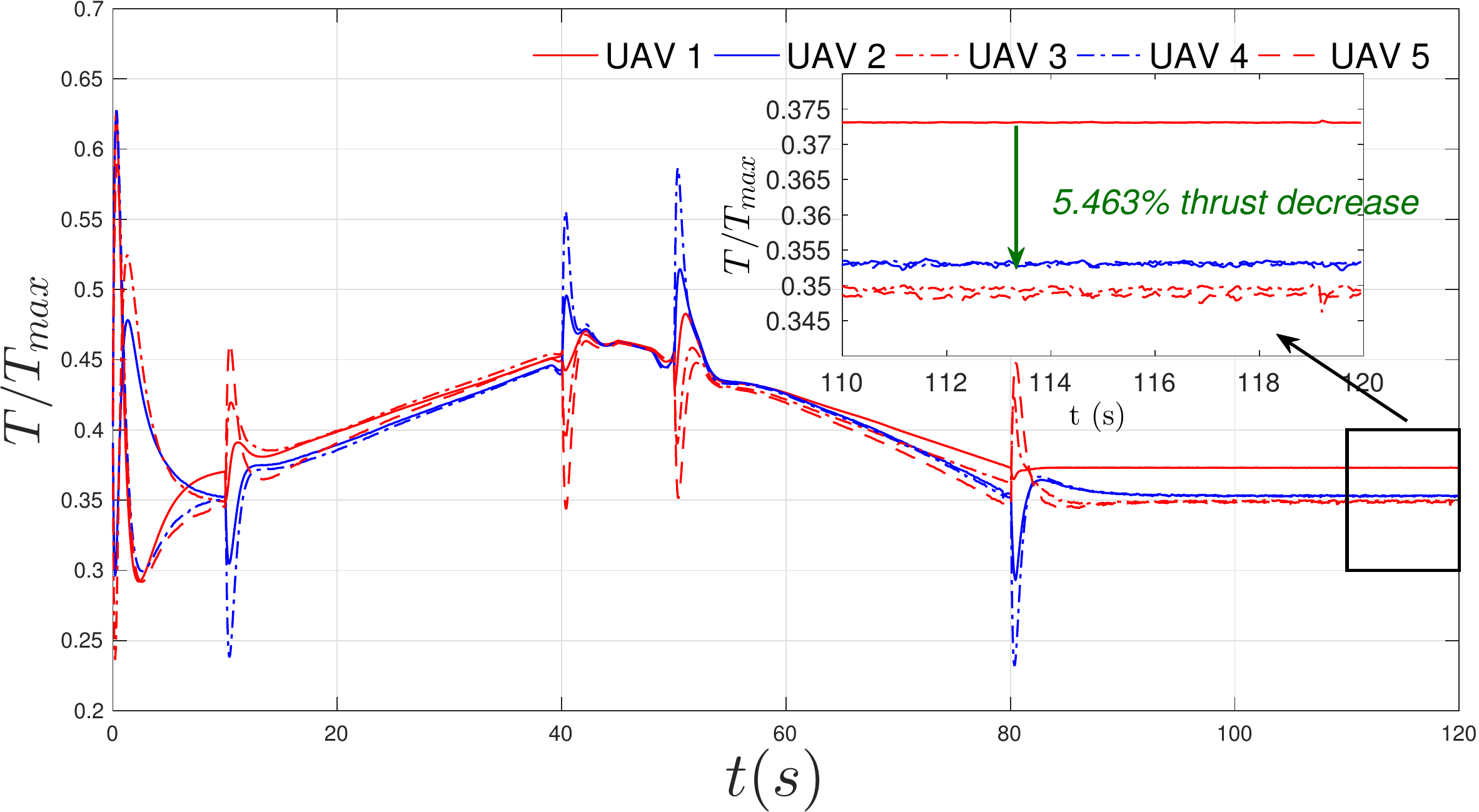}  \\ \vspace{-3mm}
\caption{Thrust inputs}
\label{Fig: VL_Thrust}
\end{minipage}\\ \vspace{-5mm}
\end{figure}
\begin{figure}[tbph]
\centering
\begin{minipage}{0.46\textwidth}
\centering
 \includegraphics[width=1\textwidth]{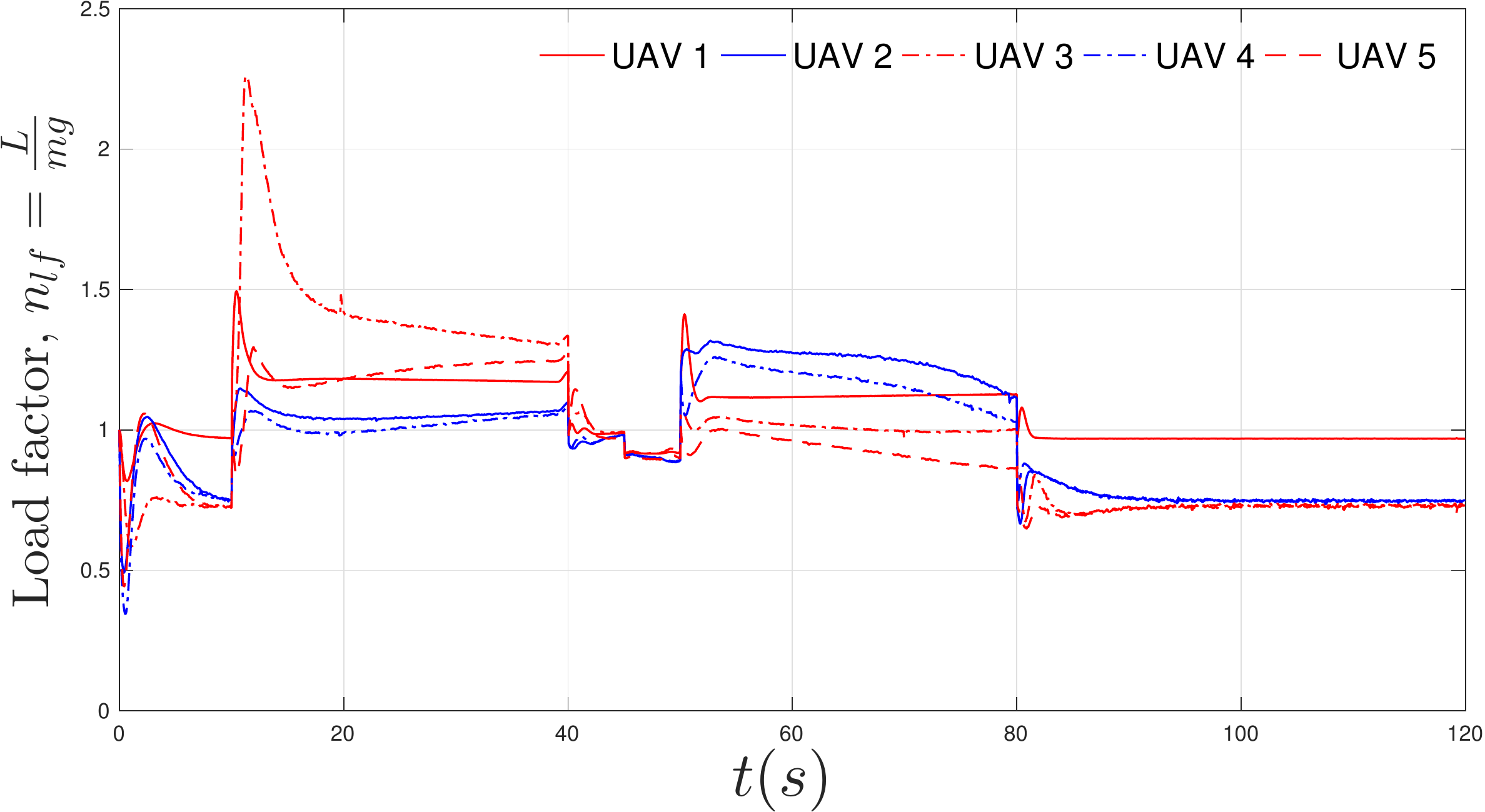} \\ \vspace{-3mm}
\caption{Load factors}
\label{Fig: VL_Lift}
\end{minipage}\hfill
\begin{minipage}{0.46\textwidth}
\centering
 \includegraphics[width=1\textwidth]{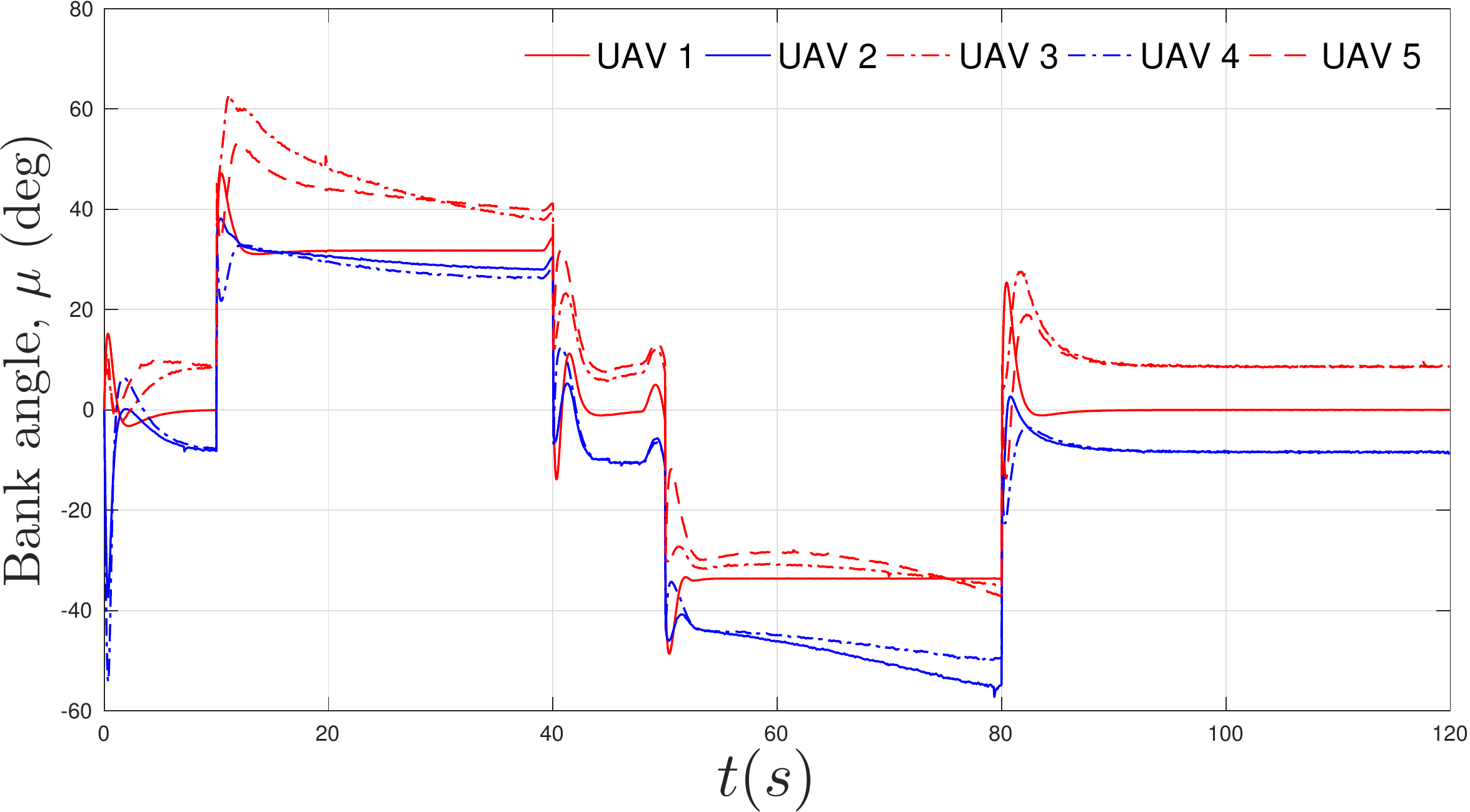} \\ \vspace{-3mm}
\caption{Bank angle inputs}
\label{Fig: VL_Mu}
\end{minipage}\\ \vspace{-5mm}
\end{figure}

\section{Conclusions} \label{Sec: Conclu}
The paper investigated the robust cooperative control for fixed-wing UAVs in close formation flight to save energy. A novel cooperative close formation controller was proposed by combining the virtual structure method and the virtual leader-based control method. The virtual structure method was introduced to describe the desired trajectories of UAVs in close formation flight. The desired trajectories were passed through a group of cooperative filters to produce the motions of virtual leaders. UAVs in close formation flight were required to track the motions of their designated virtual leaders. The model uncertainties induced by trailing vortices of other UAV were estimated and compensated by using uncertainty and disturbance observers. The analysis has shown that the states of the virtual leaders will exponentially converge to the desired formation trajectories, while the proposed robust cooperative close formation controller could at least ensure bounded close formation tracking performance. Numerical simulations on close formation flight of five UAVs were performed to show the efficiency of the proposed design.

\ifCLASSOPTIONcaptionsoff
  \newpage
\fi

\bibliographystyle{IEEETran} 
\bibliography{MyRef.bib}

\end{document}